\documentclass[lettersize,twocolumn,journal]{IEEEtran}
\usepackage{amsmath,amsfonts}
\usepackage{algorithmic}
\usepackage{algorithm}
\usepackage{array}
\usepackage[caption=false,font=normalsize,labelfont=sf,textfont=sf]{subfig}
\usepackage{textcomp}
\usepackage{stfloats}
\usepackage{url}
\usepackage{verbatim}
\usepackage{graphicx}
\usepackage{cite}
\usepackage{hyperref}
\usepackage{makecell}
\usepackage{caption}
\usepackage{color}
\usepackage{extarrows}
 \hypersetup{colorlinks,linkcolor={blue},citecolor={red},urlcolor={blue}}
\usepackage{tikz,float,multirow,booktabs,enumerate}
\usetikzlibrary{arrows, decorations.markings}
\tikzstyle{vecArrow} = [thick, decoration={markings,mark=at position
	1 with {\arrow[semithick]{open triangle 60}}},
double distance=1.4pt, shorten >= 5.5pt,
preaction = {decorate},
postaction = {draw,line width=1.4pt, white,shorten >= 4.5pt}]
\tikzstyle{innerWhite} = [semithick, white,line width=1.4pt, shorten >= 4.5pt]

\newtheorem{definition}{Definition}
\newtheorem{proposition}[definition]{Proposition}
\newtheorem{lemma}[definition]{Lemma}

\newtheorem{theorem}[definition]{Theorem}
\newtheorem{corollary}[definition]{Corollary}
\newtheorem{conjecture}[definition]{Conjecture}

\newtheorem{remark}[definition]{Remark}
\newtheorem{example}[definition]{Example}
\newtheorem{question}[definition]{Question}

\def\bcj{\begin{conjecture}}
	\def\ecj{\end{conjecture}}
\def\bcr{\begin{corollary}}
	\def\ecr{\end{corollary}}
\def\bd{\begin{definition}}
	\def\ed{\end{definition}}
\def\bea{\begin{eqnarray}}
\def\eea{\end{eqnarray}}
\def\bem{\begin{enumerate}}
	\def\eem{\end{enumerate}}
\def\bex{\begin{example}}
	\def\eex{\end{example}}
\def\bim{\begin{itemize}}
	\def\eim{\end{itemize}}
\def\bl{\begin{lemma}}
	\def\el{\end{lemma}}
\def\bma{\begin{bmatrix}}
	\def\ema{\end{bmatrix}}
\def\bpf{\begin{proof}}
	\def\epf{\end{proof}}
\def\bpp{\begin{proposition}}
	\def\epp{\end{proposition}}
\def\bqu{\begin{question}}
	\def\equ{\end{question}}
\def\br{\begin{remark}}
	\def\er{\end{remark}}

\def\squareforqed{\hbox{\rlap{$\sqcap$}$\sqcup$}}
\def\qed{\ifmmode\squareforqed\else{\unskip\nobreak\hfil
		\penalty50\hskip1em\null\nobreak\hfil\squareforqed
		\parfillskip=0pt\finalhyphendemerits=0\endgraf}\fi}
\def\endenv{\ifmmode\;\else{\unskip\nobreak\hfil
		\penalty50\hskip1em\null\nobreak\hfil\;
		\parfillskip=0pt\finalhyphendemerits=0\endgraf}\fi}
\newenvironment{proof}{\noindent \textbf{{Proof.~} }}{\qed}
\def\Dbar{\leavevmode\lower.6ex\hbox to 0pt
	{\hskip-.23ex\accent"16\hss}D}
\makeatletter
\def\url@leostyle{%
	\@ifundefined{selectfont}{\def\UrlFont{\sf}}{\def\UrlFont{\small\ttfamily}}}
\makeatother
\def\bcj{\begin{conjecture}}
	\def\ecj{\end{conjecture}}
\def\bcr{\begin{corollary}}
	\def\ecr{\end{corollary}}
\def\bd{\begin{definition}}
	\def\ed{\end{definition}}
\def\bea{\begin{eqnarray}}
	\def\eea{\end{eqnarray}}
\def\bem{\begin{enumerate}}
	\def\eem{\end{enumerate}}
\def\bex{\begin{example}}
	\def\eex{\end{example}}
\def\bim{\begin{itemize}}
	\def\eim{\end{itemize}}
\def\bl{\begin{lemma}}
	\def\el{\end{lemma}}
\def\bpf{\begin{proof}}
	\def\epf{\end{proof}}
\def\bpp{\begin{proposition}}
	\def\epp{\end{proposition}}
\def\bqu{\begin{question}}
	\def\equ{\end{question}}
\def\br{\begin{remark}}
	\def\er{\end{remark}}

\def\btb{\begin{tabular}}
	\def\etb{\end{tabular}}

	\newcommand{\nc}{\newcommand}
	
	\nc{\bbA}{\mathbb{A}} \nc{\bbB}{\mathbb{B}} \nc{\bbC}{\mathbb{C}}
	\nc{\bbD}{\mathbb{D}} \nc{\bbE}{\mathbb{E}} \nc{\bbF}{\mathbb{F}}
	\nc{\bbG}{\mathbb{G}} \nc{\bbH}{\mathbb{H}} \nc{\bbI}{\mathbb{I}}
	\nc{\bbJ}{\mathbb{J}} \nc{\bbK}{\mathbb{K}} \nc{\bbL}{\mathbb{L}}
	\nc{\bbM}{\mathbb{M}} \nc{\bbN}{\mathbb{N}} \nc{\bbO}{\mathbb{O}}
	\nc{\bbP}{\mathbb{P}} \nc{\bbQ}{\mathbb{Q}} \nc{\bbR}{\mathbb{R}}
	\nc{\bbS}{\mathbb{S}} \nc{\bbT}{\mathbb{T}} \nc{\bbU}{\mathbb{U}}
	\nc{\bbV}{\mathbb{V}} \nc{\bbW}{\mathbb{W}} \nc{\bbX}{\mathbb{X}}
	\nc{\bbZ}{\mathbb{Z}}
	
	\nc{\bA}{{\bf A}} \nc{\bB}{{\bf B}} \nc{\bC}{{\bf C}}
	\nc{\bD}{{\bf D}} \nc{\bE}{{\bf E}} \nc{\bF}{{\bf F}}
	\nc{\bG}{{\bf G}} \nc{\bH}{{\bf H}} \nc{\bI}{{\bf I}}
	\nc{\bJ}{{\bf J}} \nc{\bK}{{\bf K}} \nc{\bL}{{\bf L}}
	\nc{\bM}{{\bf M}} \nc{\bN}{{\bf N}} \nc{\bO}{{\bf O}}
	\nc{\bP}{{\bf P}} \nc{\bQ}{{\bf Q}} \nc{\bR}{{\bf R}}
	\nc{\bS}{{\bf S}} \nc{\bT}{{\bf T}} \nc{\bU}{{\bf U}}
	\nc{\bV}{{\bf V}} \nc{\bW}{{\bf W}} \nc{\bX}{{\bf X}}
	\nc{\ba}{{\bf a}} \nc{\be}{{\bf e}} \nc{\bu}{{\bf u}}
	\nc{\brr}{{\bf r}} \nc{\bx}{{\bf x}} \nc{\bz}{{\bf z}}
    \nc{\bv}{{\bf v}}  \nc{\bs}{{\bf s}}  \nc{\bt}{{\bf t}}

	\nc{\cA}{{\cal A}} \nc{\cB}{{\cal B}} \nc{\cC}{{\cal C}}
	\nc{\cD}{{\cal D}} \nc{\cE}{{\cal E}} \nc{\cF}{{\cal F}}
	\nc{\cG}{{\cal G}} \nc{\cH}{{\cal H}} \nc{\cI}{{\cal I}}
	\nc{\cJ}{{\cal J}} \nc{\cK}{{\cal K}} \nc{\cL}{{\cal L}}
	\nc{\cM}{{\cal M}} \nc{\cN}{{\cal N}} \nc{\cO}{{\cal O}}
	\nc{\cP}{{\cal P}} \nc{\cQ}{{\cal Q}} \nc{\cR}{{\cal R}}
	\nc{\cS}{{\cal S}} \nc{\cT}{{\cal T}} \nc{\cU}{{\cal U}}
	\nc{\cV}{{\cal V}} \nc{\cW}{{\cal W}} \nc{\cX}{{\cal X}}
	\nc{\cZ}{{\cal Z}}
	
	\nc{\hA}{{\hat{A}}} \nc{\hB}{{\hat{B}}} \nc{\hC}{{\hat{C}}}
	\nc{\hD}{{\hat{D}}} \nc{\hE}{{\hat{E}}} \nc{\hF}{{\hat{F}}}
	\nc{\hG}{{\hat{G}}} \nc{\hH}{{\hat{H}}} \nc{\hI}{{\hat{I}}}
	\nc{\hJ}{{\hat{J}}} \nc{\hK}{{\hat{K}}} \nc{\hL}{{\hat{L}}}
	\nc{\hM}{{\hat{M}}} \nc{\hN}{{\hat{N}}} \nc{\hO}{{\hat{O}}}
	\nc{\hP}{{\hat{P}}} \nc{\hR}{{\hat{R}}} \nc{\hS}{{\hat{S}}}
	\nc{\hT}{{\hat{T}}} \nc{\hU}{{\hat{U}}} \nc{\hV}{{\hat{V}}}
	\nc{\hW}{{\hat{W}}} \nc{\hX}{{\hat{X}}} \nc{\hZ}{{\hat{Z}}}
	
	\nc{\hn}{{\hat{n}}}
	
	\def\dim{\mathop{\rm Dim}}

	\def\ghz{\mathop{\rm GHZ}}
     \def\w{\mathop{\rm W}}

	\def\max{\mathop{\rm max}}
	\def\min{\mathop{\rm min}}

	\def\supp{\mathop{\rm supp}}
	\def\tr{\mathop{\rm Tr}}
	\def\w{\mathop{\rm W}}

	\def\wt{\mathop{\rm wt}}
	
	\newcommand{\bra}[1]{\langle#1|}
	\newcommand{\ket}[1]{|#1\rangle}
	
	\newcommand{\ketbra}[2]{|#1\rangle\!\langle#2|}
	\newcommand{\braket}[2]{\langle#1|#2\rangle}

	\newcommand{\fl}[2]{\left\lfloor\frac{#1}{#2}\right\rfloor}

\begin{document}
\begin{sloppypar}
		\title{Exploring Quantum Weight Enumerators From the $n$-Qubit  Parallelized SWAP Test}

	\author{Fei Shi, Kaiyi Guo, Xiande Zhang and Qi Zhao
	\thanks{The research of F. Shi, K. Guo and Q. Zhao was supported by the Quantum Science and Technology-National Science and Technology Major Project 2024ZD0301900, National Natural Science Foundation of China (NSFC) via Project No. 12347104, No. 12305030, and No. 12571493, Guangdong Basic and Applied Basic Research Foundation via Project 2023A1515012185, Hong Kong Research Grant Council (RGC) via No. 27300823, N\_HKU718/23, and R6010-23, Guangdong Provincial Quantum Science Strategic Initiative No. GDZX2303007. The research of  X. Zhang was supported by the Innovation Program for Quantum Science and Technology 2021ZD0302902, the NSFC under Grants No. 12171452 and No. 12231014,  and the National Key Research and Development Programs of China 2023YFA1010200 and 2020YFA0713100. (Corresponding authors:  Xiande Zhang; Qi Zhao.) }



   \thanks{Fei Shi is with Institute of Quantum Computing and Software, School of Computer Science and Engineering, Sun Yat-sen University, Guangzhou 510006, China; and with QICI Quantum Information and Computation Initiative, School of Computing and Data Science, The University of Hong Kong, Pokfulam Road, Hong Kong SAR, China. }
   
	 \thanks{Kaiyi Guo and Qi Zhao are with QICI Quantum Information and Computation Initiative, School of Computing and Data Science, The University of Hong Kong, Pokfulam Road, Hong Kong SAR, China (email: zhaoqi@cs.hku.hk).   }
	
	\thanks{Xiande Zhang is with  School of Mathematical Sciences,
		University of Science and Technology of China, Hefei, 230026, China; and with Hefei National Laboratory, University of Science and Technology of China, Hefei 230088, China (email: drzhangx@ustc.edu.cn).
	}
}

\markboth{IEEE TRANSACTIONS ON INFORMATION THEORY}%
{Shell \MakeLowercase{\textit{et al.}}: A Sample Article Using IEEEtran.cls for IEEE Journals}

\maketitle

\begin{abstract}
Quantum weight enumerators are fundamental tools for analyzing quantum error-correcting codes  and multipartite entanglement, offering insights into the existence of quantum error-correcting codes and $k$-uniform states. In this work, we establish a connection between quantum weight enumerators and the $n$-qubit parallelized SWAP test. We demonstrate that each shadow enumerator corresponds to a probability derived from this test, providing a physical interpretation for the shadow enumerators. Leveraging the non-negativity of these probabilities, we present an elegant proof for the shadow inequalities. Additionally, we show that the Shor-Laflamme weight enumerators and the Rains unitary enumerators can be calculated using the $n$-qubit parallelized SWAP test. For applications, we utilize this test to compute the distances of quantum error-correcting codes, determine the $k$-uniformity of pure states, and evaluate multipartite entanglement measures. Our results indicate that quantum weight enumerators can be efficiently estimated on quantum computers, opening a path to calculate and verify the distances of quantum error-correcting codes.
\end{abstract}


\section{Introduction}

\IEEEPARstart{S}hor and Laflamme introduced quantum weight enumerators as a method to characterize quantum error-correcting codes (QECCs) \cite{shor1997quantum}. Rains further expanded this framework by introducing two additional weight enumerators: the unitary enumerators and the shadow enumerators \cite{rains1998quantum,rains1999quantum}. For a stabilizer code, the Shor-Laflamme weight enumerators describe the weight distribution of Pauli strings in the stabilizer group \cite{shor1997quantum}, while the shadow enumerators describe the weight distribution of Pauli strings in the shadow \cite{rains1999quantum}. The Rains unitary enumerators, on the other hand, represent the sum of the subsystem purities of the QECCs \cite{rains1998quantum}.

Quantum weight enumerators can be transformed into each other through the quantum MacWilliams identities \cite{shor1997quantum,rains1998quantum,rains1999quantum,huber2018bounds}, and provide a powerful tool for investigating the existence of QECCs.
For example, they have been employed to investigate the existence of asymmetric quantum codes \cite{sarvepalli2009asymmetric}, entanglement-assisted codes \cite{lai2017linear}, hybrid codes \cite{grassl2017codes},  quantum amplitude damping codes  \cite{ouyang2020linear},  and quantum low-density parity-check (qLDPC) codes \cite{leverrier2023decoding}. Furthermore, the distance of a QECC is an important parameter, which represents its error-correcting capability and can also be characterized by the quantum weight enumerators \cite{shor1997quantum,rains1998quantum,rains1999quantum}. Recently, quantum weight enumerators have been connected to tensor networks \cite{cao2023quantum,cao2024quantum}.

Quantum weight enumerators are also related to multipartite entanglement \cite{scott2004multipartite}. An $n$-partite pure state  is a $k$-uniform state if all the reduced density matrices of $k$ parties are maximally mixed \cite{scott2004multipartite}. Specifically, a
$\fl{n}{2}$-uniform state is referred to as an absolutely maximally entangled (AME) state \cite{helwig2012absolute}, exhibiting maximal entanglement across every bipartition. The existence of
$k$-uniform states \cite{goyeneche2014genuinely,feng2017multipartite,li2019k,pang2019two,shi2022k,shi2024bounds,feng2023constructions,goyeneche2016multipartite} and AME states \cite{goyeneche2015absolutely,huber2017absolutely,horodecki2022five,rather2022thirty,raissi2020constructions,raissi2022general,shen2021absolutely}  has garnered significant attention.  The non-negativity of quantum weight enumerators provides valuable constraints, leading to non-existence results for certain
$k$-uniform and AME states \cite{rains1999quantum,scott2004multipartite,huber2018bounds,shi2022k,shi2024bounds}.
Despite their importance in the study of QECCs and multipartite entanglement, we do not yet know how to calculate quantum weight enumerators using quantum circuits.  In particular, quantum shadow enumerators have lacked an operational interpretation for over twenty years.

In this work,  we build a connection between quantum weight enumerators and the $n$-qubit  parallelized SWAP test.
The standard SWAP test can be used to calculate the overlap between two states \cite{barenco1997stabilization,buhrman2001quantum,garcia2013swap,cincio2018learning,harrow2013testing, fanizza2020beyond}, while the $n$-qubit  parallelized SWAP test can be employed to evaluate  multipartite entanglement measures~\cite{beckey2021computable,cullen2022calculating,schatzki2022hierarchy}.
The main contributions of this work are summarized as follows.

\begin{enumerate}
    \item In the $n$-qubit parallelized SWAP test of two states, we find that the probability of measuring a bitstring on the $n$ ancilla qubits can be expressed as a function of the overlaps between the reduced states of $\rho$ and $\sigma$.
Conversely, we show that the overlap between the reduced states of $\rho$ and $\sigma$ can be expressed as a function of the probabilities.

\item  We show that each probability corresponds to a shadow enumerator, which provides a computable and operational meaning for the shadow enumerators.
 Since each probability is non-negative, we also provide an alternative proof for the shadow inequalities via parallelized SWAP test circuits.

 \item  We show that both the Shor-Laflamme weight enumerators and the Rains unitary enumerators can be calculated from the probabilities in the $n$-qubit parallelized SWAP test.

 \item  For applications, we employ the
$n$-qubit parallelized SWAP test to calculate the distances of QECCs, determine the
$k$-uniformity of pure states, and evaluate multipartite entanglement measures.
\end{enumerate}

Our paper is organized as follows. In Section~\ref{pre}, we introduce the concepts of QECCs, $k$-uniform states, and quantum weight enumerators. In Section~\ref{sec:swap}, we present key properties of the
$n$-qubit parallelized SWAP test. Section~\ref{sec:connection} establishes the connection between quantum weight enumerators and the
$n$-qubit parallelized SWAP test. In Section~\ref{sec:application}, we discuss several applications of these results. Finally, we conclude the paper in Section~\ref{sec:conclusion}.

\section{Preliminaries}\label{pre}

In this section, we recall the concepts of QECCs, $k$-uniform states, quantum weight enumerators, and some related results.

\subsection{QECCs and $k$-uniform states}
We denote $[n]:=\{1, 2, \ldots, n\}$, and let $\cH_{[n]}:=\bigotimes_{i=1}^n\cH_i$  represent the Hilbert space of the $n$-partite system. Unless otherwise specified,   we always assume that the local dimension $\dim(\cH_i)=d$ for $1\leq i\leq n$. A $k$-subset of $[n]$ is a subset containing exactly $k$ elements.   A pure state $\ket{\psi}$ is a normalized column vector in $\cH_{[n]}$ (note that $\bra{\psi}=(\ket{\psi})^{\dagger}$ is a row vector, where $``\dagger"$ is the conjugate transpose operation), and a mixed state $\rho$  is a positive semidefinite matrix on $\cH_{[n]}$ with $\tr(\rho)=1$.
Let $\{e_j\}_{j\in \bbZ_{d^2}}$ be  an orthogonal  basis of matrices   acting on $\cH_i$ for $1\leq i\leq n$,  satisfying  $\tr(e_j^{\dagger}e_{j'})=\delta_{jj'}d$ and $I_d\in \{e_j\}_{j\in \bbZ_{d^2}}$, where $I_d$ is the identity matrix of order $d$. When $d=2$, the set $\{e_j\}_{j\in \bbZ_{4}}$ is chosen from the Pauli matrices:
\begin{equation}
\begin{aligned}
    I&=\begin{pmatrix}
        1 & 0 \\
        0 & 1
    \end{pmatrix}, \,
    X=\begin{pmatrix}
        0 & 1 \\
        1 & 0
    \end{pmatrix}, \,\\
    Y&=\begin{pmatrix}
        0 & -i \\
        i & 0
    \end{pmatrix}, \,
        Z=\begin{pmatrix}
        1 & 0 \\
        0 & -1
    \end{pmatrix}.
\end{aligned}
\end{equation}
When $d\geq 3$, the set $\{e_j\}_{j\in \bbZ_{d^2}}$ can be chosen from the generalized Pauli matrices \cite{bertlmann2008bloch}.  An error basis $\cE_d:=\{E_{\alpha}\}_{\alpha\in \bbZ_{d^2}^n}$ acting on  $\cH_{[n]}$ consists of
$E_{\alpha}=e_{\alpha_1}\otimes e_{\alpha_2}\otimes\cdots\otimes e_{\alpha_n}$, where each $e_{\alpha_i}$ acts on $\cH_i$ for $1\leq i\leq n$, and $\alpha=(\alpha_1, \alpha_2, \ldots, \alpha_n)\in \bbZ_{d^2}^n$. The support of $E_{\alpha}$ is defined as
 $\supp(E_\alpha):=\{i\mid e_{\alpha_i}\neq I_d, \, \forall \, 1\leq i\leq n\}$, and the weight of $E_{\alpha}$ is the size of $\supp(E_\alpha)$, i.e.,  $\wt(E_{\alpha}):=|\supp(E_\alpha)|$. Let us review QECCs  \cite{knill1997theory,scott2004multipartite}.
\begin{definition}
    Let $\cQ$ be a subspace of $\cH_{[n]}$ with dimension $K$, then $\cQ$ is an \emph{$((n, K, \delta))_d$  QECC} if for all states $\ket{\psi}\in \cQ$, and errors $E_{\alpha}\in \cE_d$ with $\wt(E_{\alpha})<\delta$,
\begin{equation}
 \bra{\psi}E_{\alpha} \ket{\psi}=C(E_{\alpha}),
\end{equation}
where  $C(E_{\alpha})$ is a constant which depends only on $E_{\alpha}$, and $\delta$ is called the \emph{distance} of the code. In particular, $\cQ$ is called \emph{pure} if $C(E_{\alpha})=\frac{\tr(E_{\alpha})}{d^n}$, and  one dimensional $((n, 1, \delta))_d$ QECCs refer only to pure QECCs.
\end{definition}

The distance of a QECC characterizes its error-correcting capability.
An $((n, K, \delta))_d$ QECC can detect all the errors acting on at most $\delta-1$ subsystems, and correct all the errors acting on at most $\fl{\delta-1}{2}$ subsystems.

Let us review a special class of QECCs---stabilizer codes \cite{calderbank1997quantum,gottesman1996class}. The
$n$-qubit Pauli group $\cP_n$ consists of all $n$-fold tensor products of the Pauli matrices, along with multiplicative factors  $\pm 1, \pm i$., i.e., $\cP_n:=\{i^\lambda\cE_2\mid \lambda\in \bbZ_4\}$. The natural group
homomorphism is $f: \cP_n\rightarrow \cE_2$, $i^\lambda E_{\alpha}\rightarrow E_{\alpha} $.
A subgroup  $G$ of  $\cP_n$ is called  a \emph{stabilizer} group if $G$ is an abelian group, and $-I_{2^n}\notin G$.
The \emph{centralizer} of the  stabilizer $G$ is defined as $C(G):=\{s\in \cP_n\mid gs=sg, \, \forall g\in G\}$.
\begin{definition}
    Assume $G$ is a stabilizer of $\cP_n$ with $n-k$ independent generators, then the subspace $\text{Fix}(G):=\{\ket{\psi}\in \cH_{[n]}\mid g\ket{\psi}=\ket{\psi}, \, \forall\, g\in G \}$ is called a \emph{stabilizer code} with parameters  $((n,2^k,\delta))_2$, where
    \begin{equation}
        \delta=\min\{\wt(g)\mid g\in f(C(G))\setminus f(G)\}.
    \end{equation}
In particular, if $\delta=\min\{\wt(g)\mid g\in f(C(G))\setminus \{I_{2^n}\}\}$, then $\text{Fix}(G)$ is pure.
\end{definition}

For example, the five-qubit code \cite{knill2001benchmarking} is a $((5,2,3))_2$ stabilizer code with stabilizer generators: $g_1=XZZXI$, $g_2=IXZZX$, $g_3=XIXZZ$, and $g_4=ZXIXZ$, where the tensor product symbol  ``$\otimes$" is omitted for simplicity.
The centralizer generators are: $g_1$, $g_2$, $g_3$, $g_4$, $XXXXX$, and $ZZZZZ$. One can verify that $3=\min\{\wt(g)\mid g\in f(C(G))\setminus \{I_{2^5}\}\}$, which means that the $((5,2,3))_2$ stabilizer code is a pure code.

Next, we introduce $k$-uniform states \cite{scott2004multipartite}. For a subset $S\subseteq [n]$, we denote $S^c:=[n]\setminus S$. For a state $\rho$ on $\cH_{[n]}$, the reduced stated on $S$ is defined as: $\rho_S:=\tr_{S^c}\rho$, where $\tr_{S^c}$ is the partial trace operation.

\begin{definition}
    A pure state $\ket{\psi}\in \cH_{[n]}$ is called a \emph{$k$-uniform state} if $\rho_{S}=\frac{1}{d^{k}}I_{d^k}$ for all $k$-subsets $S$ of $[n]$, where $\rho=\ketbra{\psi}{\psi}$. In particular, a $\fl{n}{2}$-uniform state is called an \emph{AME} state.
\end{definition}

For example,  the $n$-qudit Greenberger-Horne-Zeilinger (GHZ) state:  $\ket{\ghz_n^d}=\frac{1}{\sqrt{d}}\sum_{i\in\bbZ_d}\ket{i}^{\otimes n}$ is a $1$-uniform state. If  $n=2$ or $3$, the GHZ state is an AME state. It is known that $k$-uniform states are connected to QECCs.

\begin{lemma}[\cite{scott2004multipartite,huber2020quantum}]\label{lemma:k_uniform}
    A subspace $\cQ$ of $\cH_{[n]}$ with dimension $K$ is a pure $((n,K,\delta))_d$ QECC if and only if $\ket{\psi}$ is a $(\delta-1)$-uniform state for all $\ket{\psi}\in \cQ$.
\end{lemma}

\subsection{Shor-Laflamme weight enumerators and Rains unitary enumerators}

 For a subspace $\cQ$ of $\cH_{[n]}$ with dimension $K$, we denote $\Pi_\cQ$ as the projector on $\cQ$, and $\rho_{\cQ}$ as the normalized projector on $\cQ$, i.e., $\rho_{\cQ}=\frac{\Pi_\cQ}{K}$.
 The (unnormalized) \emph{Shor-Laflamme} weight  enumerators \cite{shor1997quantum} are defined as:
\begin{equation*}
\begin{aligned}
    A_j&:=\sum_{\wt(E_\alpha)=j, \,  E_\alpha\in \cE_d}\tr(E_{\alpha}\rho_{\cQ})\tr(E_{\alpha}^{\dagger}\rho_{\cQ});\\
    \end{aligned}
\end{equation*}
\begin{equation}
\begin{aligned}
    B_j&:=\sum_{\wt(E_\alpha)=j, \,  E_\alpha\in \cE_d}\tr(E_{\alpha}\rho_{\cQ} E_{\alpha}^{\dagger}\rho_{\cQ}).
\end{aligned}
\end{equation}
 The \emph{Rains unitary enumerators} \cite{rains1998quantum} are defined as:
\begin{equation}
\begin{aligned}
    A_j'&:=\sum_{|T|=j,  \, T\subseteq [n]} \tr[(\rho_\cQ)_T^2];\\
    B_j'&:=\sum_{|T|=j,  \, T\subseteq [n]} \tr[(\rho_\cQ)_{T^c}^2].
\end{aligned}
\end{equation}

Note that $\tr[(\rho_\cQ)_T^2]$ is called the \emph{purity} of the reduced state $(\rho_{\cQ})_T$.
The Shor-Laflamme weight enumerators and the Rains unitary enumerators  can be used to characterize the distances of QECCs.

\begin{lemma}[\cite{rains1998quantum,scott2004multipartite}]\label{lemma:distance}
    A subspace $\cQ$ of $\cH_{[n]}$ with dimension $K$ is an  $((n, K, \delta))_d$ QECC  if and only if
$KB_j= A_j$ for $0< j< \delta$ (or $KB_{\delta-1}'= A_{\delta-1}'$). In particular, $\cQ$ is a pure  $((n, K, \delta))_d$ QECC  if and only if $B_j= A_j=0$ for $0< j< \delta$ (or $KB_{\delta-1}'= A_{\delta-1}'=\binom{n}{\delta-1}d^{1-\delta}$).
\end{lemma}

For an $((n,2^k,\delta))_2$ stabilizer code with stabilizer $G$, $A_j$ counts the number of elements with weight $j$ in $f(G)$, and $2^kB_j$ counts the number of elements with weight $j$ in $f(C(G))$ \cite{shor1997quantum,rains1999quantum}.  A combinatorial interpretation of the Shor-Laflamme weight enumerators for CWS codes is provided in Ref.~\cite{9738648}.
 The Rains unitary enumerators can be understood as
the sum of the subsystem purities of the code.

\subsection{Shadow enumerators}

Let $\rho$ and $\sigma$ be two states on $\cH_{[n]}$ (the local dimension does not necessarily have to be two or even equal), and $T$ be an arbitrary subset of $[n]$,  the (normalized) \emph{shadow enumerator} \cite{rains1998quantum} of $\rho$ and $\sigma$ on $T$ is defined as:
\begin{equation}\label{eq:st}
    s_{T}(\rho,\sigma):=\frac{1}{2^n}\sum_{S\subseteq [n]}(-1)^{|S\cap T^c|}\tr(\rho_{S}\sigma_{S}),
\end{equation}
and the \emph{shadow inequalities} \cite{rains2000polynomial} are:
\begin{equation}
    s_{T}(\rho,\sigma)\geq 0, \quad \forall\, T\subseteq [n].
\end{equation}
When $\rho=\sigma$ is an $n$-qubit state, the  shadow enumerator has another definition \cite{rains1999quantum}:
\begin{equation}
    s_{T}(\rho,\rho):=\frac{1}{2^n}\sum_{\supp(E_{\alpha})=T,\, E_{\alpha}\in \cE_2}\tr(\rho E_{\alpha} \widetilde{\rho} E_{\alpha}),
\end{equation}
where $\widetilde{\rho}=Y^{\otimes n}\overline{\rho}Y^{\otimes n}$, and $\overline{\rho}$ is the conjugate of $\rho$. Note that $\rho\longmapsto \widetilde{\rho}$ is called the \emph{state inversion map} \cite{hall2005multipartite,wyderka2018constraints,eltschka2018distribution}. In Appendix~\ref{appendix:lemma_equivalent}, we show that these two definitions of shadow enumerators are equivalent when $\rho=\sigma$ is an $n$-qubit state.

For an $((n,2^k,\delta))_2$ stabilizer code $\cQ$ with stabilizer $G$, we denote $G_0=\{g\in G\mid \text{$\wt(g)$ is even}\}$ ($G_0$ is a subgroup of $G$). The \emph{shadow} of $G$ is defined as \cite{rains1999quantum}:
\begin{equation}
S(G):=\begin{cases}
    f(C(G)), & \text{$\wt(g)$ is even $\forall \, g\in G$;}   \\
f(C(G_0))\setminus f(C(G)),
& \text{$\wt(g)$ is odd $\exists \, g\in G$}.
\end{cases}
\end{equation}
Let $s_j=\sum_{|T|=j,\, T\subseteq [n]}s_T(\rho_\cQ,\rho_\cQ)$, then $2^{n+k}s_j$ counts the number of elements with weight $j$ in the shadow $S(G)$ \cite{rains1999quantum}.

The shadow inequalities play an important role in quantum information theory. For example, the shadow inequalities give $2^{n-1}-1$ monogamy inequalities (see Appendix~\ref{appendix: monogamy} for the details):

\begin{equation}
   \sum_{ S\subseteq [n]}(-1)^{|S\cap T|}C_{S|S^c}^2(\ket{\psi})\leq 0,  \  |T| \ \text{is even}, \ \emptyset\neq T\subseteq [n],
\end{equation}
where $\ket{\psi}\in \cH_{[n]}$,  $C_{S|S^c}(\ket{\psi})=\sqrt{2\left[1-\tr(\rho_S^2)\right]}$ is the \emph{concurrence} \cite{wootters1998entanglement}, and $\rho=\ketbra{\psi}{\psi}$.
In particular, if $n=3$ and $T=\{1,2\}$, we would obtain the triangle inequality \cite{zhu2015generalized,yang2022entanglement}:
  $C_{3|12}^2(\ket{\psi})\leq  C_{1|23}^2(\ket{\psi})+ C_{2|13}^2(\ket{\psi})$. The shadow inequalities can also be used to show some non-existence results for AME states (i.e., pure $((n,1,\fl{n}{2}+1))_d$ QECCs)\cite{gisin1998bell,higuchi2000entangled,helwig2012absolute,gour2010all,scott2004multipartite,huber2018bounds}. For example, assume $\ket{\psi}$ is a $4$-qubit AME state. In the shadow inequalites, let $\rho=\sigma=\ketbra{\psi}{\psi}$, and $T=\emptyset$, then
\begin{equation}
\begin{aligned}
     \sum_{S\subseteq [4]}(-1)^{|S\cap T^c|}\tr(\rho_S^2)=&1-\binom{4}{1}\times\frac{1}{2}+\binom{4}{2}\times\frac{1}{4}\\&-\binom{4}{3}\times\frac{1}{2}+1=-\frac{1}{2}<0.
\end{aligned}
\end{equation}
Hence, $4$-qubit AME states do not exist \cite{higuchi2000entangled,huber2018bounds}.

The Shor-Laflamme weight  enumerators, Rains unitary enumerators, and shadow enumerators can be transformed into one another through the quantum MacWilliams identities \cite{shor1997quantum,rains1998quantum,rains1999quantum,huber2018bounds}. For QECCs, these quantum weight enumerators are all non-negative, and they can be used to give some linear programming bounds for various types of QECCs \cite{shor1997quantum,rains1998quantum,lai2017linear,grassl2017codes,ouyang2020linear}. For example, the famous bound for an $((n,K,\delta))_d$ QECC is the quantum Singltion bound: $K\leq d^{n-2\delta+2}$, which can be also obtained from the non-negativity of the Shor-Laflamme weight enumerators and the Rains unitary enumerators \cite{rains1999nonbinary}. In this work, we focus on calculating these quantum weight enumerators using quantum circuits.

\section{The $n$-qubit  parallelized SWAP test}\label{sec:swap}
In this section, we introduce the $n$-qubit  parallelized SWAP test.
For a bitstring $\bz=(z_1, z_2, \ldots, z_n)\in \bbZ_2^n$, we denote $\supp(\bz):=\{i\mid z_i\neq 0, \, \forall \, 1\leq i\leq n\}$, and $\wt(\bz):=|\supp(\bz)|$. Let $\bu=(u_1, u_2, \ldots, u_n), \bv=(v_1, v_2, \ldots, v_n)\in \bbZ_2^n$, then the inner product of $\bu$ and $\bv$ is $\bu\cdot \bv:=\sum_{i=1}^nu_iv_i=|\supp(\bu)\cap \supp(\bv)|$. Furthermore, the \emph{overlap} between two $n$-partite states $\rho$ and $\sigma$ is defined as $\tr(\rho\sigma)$.

Ref.~\cite{beckey2021computable} provides a useful tool for analyzing quantum weight enumerators---the
$n$-qubit parallelized SWAP test, and our central idea is based on this work.
Let $\rho$ and $\sigma$ be  two  states on $\cH_{[n]}$ (the local dimension does not necessarily have to be two or even equal), then the $n$-qubit parallelized SWAP test for $\rho$ and $\sigma$ is given in  Fig.~\ref{fig:swap}. We denote $p(\bz)$ as the probability of measuring bitstring $\bz\in \bbZ_2^n$  on the $n$ ancilla  qubits. It is known that $p(\bz)=|\braket{\lambda}{\bz}|^2$ when measuring bitstring $\bz\in \bbZ_2^n$ on an $n$-qubit state $\ket{\lambda}$. Then we can express $p(\bz)$  as a function of the  overlaps between the reduced states of $\rho$ and $\sigma$.

\begin{proposition}\label{proposition: pz}
Let  $\rho$ and $\sigma$ be two  states on $\cH_{[n]}$. In the $n$-qubit parallelized SWAP test of $\rho$ and $\sigma$, we have
\begin{equation}
p(\bz)= \frac{1}{2^n}\sum_{S\subseteq [n]}(-1)^{|S\cap T|}\tr(\rho_S\sigma_S),
\end{equation}
where $\bz\in \bbZ_2^n$, and $T=\supp(\bz)$.
\end{proposition}
\begin{figure}[t]
		\centering
		\includegraphics[scale=0.35]{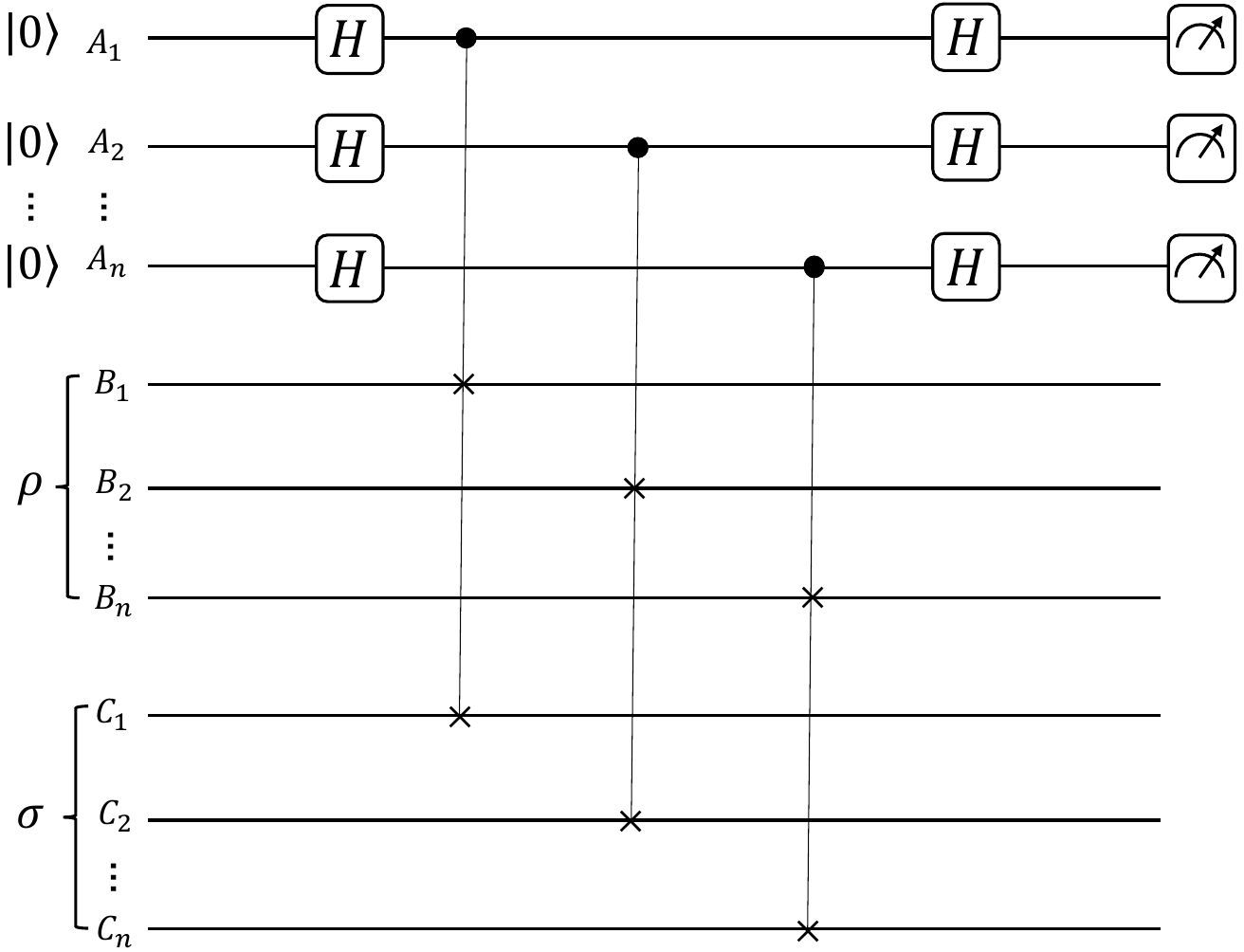}
		\caption{Quantum circuit for the $n$-qubit parallelized SWAP test of two $n$-partite states $\rho$ and $\sigma$. Each $A_i$ is an ancilla qubit,   each $H$ is a Hadamard gate, and each controlled-SWAP gate is performed on $A_i$, $B_i$, and $C_i$.    } \label{fig:swap}
\end{figure}

The proof of Proposition~\ref{proposition: pz}
is given in Appendix~\ref{appendix: proposition 1}.
Proposition~\ref{proposition: pz} extends the results of Ref.~\cite{beckey2021computable}, which only considers the case where $\rho=\sigma$ is an $n$-qubit pure state. We present some examples for the $n$-qubit parallized SWAP test with some equal pure states and some non-equal pure states in Table~\ref{Table: probability}.

The overlap between two $n$-partite states $\rho$ and $\sigma$, i.e., $\tr(\rho\sigma)$  can be computed from the standard SWAP test \cite{barenco1997stabilization,buhrman2001quantum,garcia2013swap,cincio2018learning,harrow2013testing, fanizza2020beyond}. This means that $2^n$ standard SWAP tests are needed to compute the  overlap  $\tr(\rho_T\sigma_T)$ for all $T\subseteq [n]$. However,  we show that only one $n$-qubit parallelized SWAP test is enough to compute all the overlaps $\tr(\rho_T\sigma_T)$.

\begin{proposition}\label{proposition: overlap}
   Let  $\rho$ and $\sigma$ be two states on $\cH_{[n]}$, and  $T\subseteq [n]$. In the $n$-qubit parallelized SWAP test of $\rho$ and $\sigma$, we have
   \begin{equation}
   \tr(\rho_T\sigma_T)=\sum_{\bz\in \bbZ_2^n}(-1)^{\bz\cdot \textbf{t}}p(\bz),
   \end{equation}
where $\bt\in \bbZ_2^n$,  and $\supp(\bt)=T$.
\end{proposition}

The proof of Proposition~\ref{proposition: overlap} is given in Appendix~\ref{appendix: proposition 2}.
Note that the probability distribution of the $n$-qubit parallelized SWAP test of $\rho$ and $\sigma$  can also be used to characterize $\rho$ and $\sigma$ to some extent. For example,
     $p(0)=1$ if and only if $\rho=\sigma=\bigotimes_{i=1}^n\ketbra{\psi_i}{\psi_i}$;
 $p(\bz)=\frac{1}{2^n}$ for all $\bz\in \bbZ_2^n$ if and only if $\rho_T\sigma_T=\mathbf{0}$ for all $\emptyset\neq T\subseteq [n]$ (see Appendix~\ref{appendix:charaterizing}  for details).

\section{Connection between quantum weight enumerators and the $n$-qubit parallelized SWAP test}\label{sec:connection}

Now, we are in a position to present our main results, which establish a connection between quantum weight enumerators and the
$n$-qubit parallelized SWAP test.

\begin{theorem}\label{theorem: shadow}
  Let $\rho$ and $\sigma$ be two states on $\cH_{[n]}$, and $T\subseteq [n]$. In the $n$-qubit parallelized SWAP test of $\rho$ and $\sigma$,  we have
  \begin{equation}
    s_{T}(\rho,\sigma)=p(\bz),
\end{equation}
where $\bz\in \bbZ_2^n$, and $\supp(\bz)=T^c$.
\end{theorem}

Theorem~\ref{theorem: shadow} is obtained from Proposition~\ref{proposition: pz} and Eq.~\eqref{eq:st}. According to Theorem~\ref{theorem: shadow},  we
give a computable and operational meaning for the shadow enumerators, i.e., the shadow enumerator of $\rho$ and $\sigma$ on $T$ is the probability of measuring the bitstring $\bz\in \bbZ_2^n$ with $\supp(\bz)=T^c$  in the  $n$-qubit parallelized SWAP test of $\rho$ and $\sigma$. The shadow inequalities were a conjecture given in Ref.~\cite{rains1998quantum}, and proven in Ref.~\cite{rains2000polynomial}.
Since $p(\bz)\geq 0$, we  provide an elegant proof for the shadow inequalities. Note that  Table~\ref{Table: probability} also lists the shadow enumerators of some equal pure states and  non-equal pure
states.

Based on Proposition~\ref{proposition: overlap} and the quantum MacWilliams identities \cite{shor1997quantum,rains1998quantum,rains1999quantum,huber2018bounds},
we can express both the Shor-Laflamme weight enumerators and the Rains unitary enumerators as  functions of the probabilities in the $n$-qubit parallelized SWAP test.
\begin{theorem}\label{theorem: enumerators}
Let $\cQ$ be an $((n,K,\delta))_d$ QECC. In the $n$-qubit parallelized SWAP test of $\rho_{\cQ}$ and $\rho_{\cQ}$, we have
  \begin{equation}
 \begin{aligned}
     A_j=&\sum_{k=0}^j(-1)^{j-k}\binom{n-k}{j-k}d^k \sum_{\wt(\bt)=k,\,  \bt\in \bbZ_2^n}\sum_{\bz\in \bbZ_2^n}(-1)^{\bz\cdot \textbf{t}}p(\bz);\\
         B_j=&\sum_{k=0}^j(-1)^{j-k}\binom{n-k}{j-k}d^k \sum_{\begin{subarray}{c}
   \wt(\bt)=n-k,\\  \bt\in \bbZ_2^n
     \end{subarray}} \sum_{\bz\in \bbZ_2^n}(-1)^{\bz\cdot \textbf{t}}p(\bz);\\
 A_j'=&\sum_{\wt(\bt)=j,\,  \bt\in \bbZ_2^n} \sum_{\bz\in \bbZ_2^n}(-1)^{\bz\cdot \textbf{t}}p(\bz);\\
        B_j'=&\sum_{\wt(\bt)=n-j,\,  \bt\in \bbZ_2^n} \sum_{\bz\in \bbZ_2^n}(-1)^{\bz\cdot \textbf{t}}p(\bz).
     \end{aligned}
    \end{equation}
    \end{theorem}
\begin{proof}
By Proposition~\ref{proposition: overlap}, we have
\begin{equation}
A_j'=\sum_{|T|=j,  \, T\subseteq [n]} \tr[(\rho_\cQ)_T^2]=\sum_{\wt(\bt)=j, \,  \bt\in \bbZ_2^n} \sum_{\bz\in \bbZ_2^n}(-1)^{\bz\cdot \textbf{t}}p(\bz).
\end{equation}
According to  the quantum MacWilliams identities \cite{shor1997quantum,rains1998quantum,rains1999quantum,huber2018bounds}, we have $B_j'=A_{n-j}'$,  $A_j=\sum_{k=0}^j(-1)^{j-k}\binom{n-k}{j-k}d^k A_k'$,  and  $B_j=\sum_{k=0}^j(-1)^{j-k}\binom{n-k}{j-k}d^k B_k'$. This completes the proof.
\end{proof}

\begin{table}[t]	
\renewcommand\arraystretch{1.7}	
	\centering
	\renewcommand\tabcolsep{2pt}
     \caption{The
probability distribution of the $n$-qubit parallelized SWAP
test of $\rho$ and $\sigma$, where $\rho, \sigma\in \{\ket{0_n}=\ket{0}^{\otimes n}, \ket{1_n}=\ket{1}^{\otimes n}, \ket{\w_n}=\frac{1}{\sqrt{n}}\sum_{\wt(\bv)=1,\, \bv\in \bbZ_2^n}\ket{\bv}, \ket{\ghz_n}=\frac{1}{\sqrt{2}}(\ket{0_n}+\ket{1_n})\}$. We denote $p(k)=p(\bz)$ for any $\wt(\bz)=k$, $\bz\in \bbZ_2^n$.
 }\label{Table: probability}
	\begin{tabular}{|c|c|c|c|c|c|}
	\hline
	 $\rho$ &$\sigma$  &$p(0)$ &$p(1)$ &$p(2)$    &$p(k), 3\leq k\leq n$ \\
		\hline
  $\ket{0_n}$ &$\ket{0_n}$ &$1$ &$0$ &$0$ &$0$ \\
 $\ket{\w_n}$ &$\ket{\w_n}$ &$\frac{n+1}{2n}$ &$0$  &$\frac{1}{n^2}$ &$0$ \\
  $\ket{\ghz_n}$ &$\ket{\ghz_n}$ &$\frac{2^{n-1}+1}{2^n}$ &$0$  &$\frac{1}{2^n}$ &$\frac{1+(-1)^k}{2^{n+1}}$ \\
$\ket{0_n}$ &$\ket{1_n}$ &$\frac{1}{2^n}$ &$\frac{1}{2^n}$ &$\frac{1}{2^n}$ &$\frac{1}{2^n}$\\
$\ket{\w_n}$ &$\ket{0_n}$ &$\frac{1}{2}$ &$\frac{1}{2n}$ &$0$ &$0$ \\
$\ket{\w_n}$ &$\ket{1_n}$ &$\frac{1}{2^{n-1}}$ &$\frac{n-1}{2^{n-1}n}$ &$\frac{n-2}{2^{n-1}n}$ &$\frac{n-k}{2^{n-1}n}$\\
 $\ket{\ghz_n}$ &$\ket{0_n}$ &$\frac{2^n+1}{2^{n+1}}$ &$\frac{1}{2^{n+1}}$ &$\frac{1}{2^{n+1}}$ &$\frac{1}{2^{n+1}}$ \\
 $\ket{\ghz_n}$ &$\ket{\w_n}$ &$\frac{2^{n-2}+1}{2^n}$ &$\frac{2^{n-2}+n-1}{2^nn}$ &$\frac{n-2}{2^{n}n}$ &$\frac{n-k}{2^nn}$\\
\hline
	\end{tabular}
\end{table}

Since the $n$-qubit parallelized SWAP test can be implemented on a quantum computer \cite{buhrman2001quantum,garcia2013swap,cincio2018learning,harrow2013testing}, we can efficiently estimate these quantum weight enumerators through Theorem~\ref{theorem: shadow} and Theorem~\ref{theorem: enumerators}.
The complexity of estimating the quantum weight enumerators can be analyzed as follows:

1. For the shadow enumerators, $s_T(\rho,\rho)=p(\mathbf{z})$, where $\mathrm{supp}(\bz)=T^{c}$. So the corresponding unbiased estimator is $\hat{s}_{T}(\rho,\rho)=\frac{N_{\mathbf{z}}}{N}$, where $N_{\mathbf z}$ is the number of occurrences of $\mathbf z$ in the sample.  This gives the variance of the estimator $\mathrm{Var} \hat{s}_{T}(\rho,\rho)=\frac{1-s_{T}(\rho,\rho)^{2}}{N}$. To reach accuracy $\epsilon$ we require sampling times $N=O\left(  \frac{1-s_{T}(\rho,\rho)^{2}}{\epsilon ^{2}} \right)$.

2. For the Rains unitary enumerators, $A_{j}'= \sum_{\mathbf{z}\in \mathbb{Z}_{2}^{n}} f_{j}(\mathrm{wt}(\mathbf{z}))p(\mathbf{z})$ and $B_{j}'=\sum_{\mathbf{z}\in \mathbb{Z}_{2}^{n}}f_{n-j}(\mathrm{wt}(\mathbf{z}))p(\mathbf{z})$. The corresponding unbiased estimator $\hat{A}_{j}'$ and $\hat{B}_{j}'$ has variance $\frac{1}{N} \left[\sum_{\mathbf{z}\in \mathbb{Z}_{2}^{n}}f_{j}(\mathrm{wt}(\mathbf{z}))^{2} p (\mathbf{z}) - A_{j}'^{2}\right]$ and $\frac{1}{N}\left[\sum_{\mathbf{z}\in \mathbb{Z}_{2}^{n}}f_{n-j}(\mathrm{wt}(\mathbf{z}))^{2}p(\mathbf{z}) - B_{j}'^{2}\right]$, respectively. To reach accuracy $\epsilon$ we require $N=O\left(  \frac{\sum_{\mathbf{z}\in \mathbb{Z}_{2}^{n}}f_{j}(\mathrm{wt}(\mathbf{z}))^{2} p (\mathbf{z}) - A_{j}'^{2}}{\epsilon ^{2}} \right)$ and $N=O\left(  \frac{\sum_{\mathbf{z}\in \mathbb{Z}_{2}^{n}}f_{n-j}(\mathrm{wt}(\mathbf{z}))^{2}p(\mathbf{z}) - B_{j}'^{2}}{\epsilon ^{2}} \right)$ sampling times, respectively.

 3. For the Shor-Laflamme weight enumerators, we have $A_{j}= \sum_{k=0}^{j}(-1)^{ j-k} {n-k\choose j-k} d^{k}A_{k}'= \sum_{\mathbf{z}\in \mathbb{Z}_{2}^{n}}g_{j}(\mathrm{wt}(\mathbf{z}))$ and $B_{j}= \sum_{k=0}^{j}(-1)^{j-k}{n-k \choose j-k} d^{k} B_{k}' = \sum_{\mathbf{z}}h_{j}(\mathrm{wt}(\mathbf{z}))p(\mathbf{z})$, where $g_j,h_j$ are some coefficients. We can find the variance of their unbiased estimators and estimate the required sampling times using similar techniques.

\section{applications}\label{sec:application}
In this section, we present several applications, including calculating the distances of QECCs, determining the
$k$-uniformity of pure states, and evaluating multipartite entanglement measures.

\subsection{Calculating the distances of QECCs}
The first application involves calculating the distances of  QECCs using the  $n$-qubit parallelized SWAP test. Determining the distances of QECCs is known to be a challenging problem. However, by leveraging Lemma~\ref{lemma:distance} and Theorem~\ref{theorem: enumerators}, we propose a method to address this difficulty.

Let $\cQ$ be a subspace of $\cH_{[n]}$ with dimension $K$. The procedure consists of the following four steps:
\begin{enumerate}[1.]
    \item prepare two copies of normalized projector $\rho_{\cQ}$;
    \item apply the $n$-qubit parallelized SWAP test of $\rho_{\cQ}$ and $\rho_{\cQ}$  to  obtain the probability distribution (shadow enumerators);
    \item compute the  Rains unitary enumerators (we only consider the Rains unitary enumerators for simplicity) by using Theorem~\ref{theorem: enumerators};
    \item determine the distance $\delta=\max\{j+1\mid KB_j'-A_j'=0\}$ from Lemma~\ref{lemma:distance}.
\end{enumerate}

The complexity of this algorithm is analyzed in the Appendix~\ref{appendix:complexity}. Note that the algorithm requires
$3n$ qubits,
$2n$ Hadamard gates, and
$n$ controlled-SWAP gates.
Next, we show an example. Let $\rho_{\cQ_5}$ be   a  5-qubit normalized projector with dimension $\dim(\cQ_5)=2$.  Assume that we obtain the probability distribution   from the $5$-qubit parallelized SWAP test of $\rho_{\cQ_5}$ and $\rho_{\cQ_5}$: $p(0)=\frac{9}{32}$, $p(1)=\frac{3}{64}$,  $p(2)=\frac{3}{64}$, $p(3)=0$, $p(4)=0$, and  $p(5)=\frac{1}{64}$, where  $p(k)=p(\bz)$ for any $\wt(\bz)=k$, $\bz\in \bbZ_2^5$. By computing the Rains unitary enumerators, we have $A'_0=1$, $A_1'=\frac{5}{2}$, $A_2'=\frac{5}{2}$, $A_3'=\frac{5}{4}$, $A_4'=\frac{5}{4}$, $A_5'=\frac{1}{2}$; and  $B_0'=\frac{1}{2}$, $B_1'=\frac{5}{4}$, $B_2'=\frac{5}{4}$, $B_3'=\frac{5}{2}$, $B_4'=\frac{5}{2}$, $B_5'=1$. Since   $2B_2'-A_2'=0$, and $2B_3'-A_3'\neq 0$,  we can determine that the distance of $\cQ_5$ is $3$.
 Actually, $\cQ_5$ is chosen from the five-qubit code \cite{knill2001benchmarking}.
See Appendix~\ref{appendix: probability} for the complete quantum weight enumerators of the five-qubit code $((5,2,3))_2$, the Steane code $((7,2,3))_2$ \cite{steane1996multiple} and the Shor code $((9,2,3))_2$ \cite{shor1995scheme}.

When preparing a QECC in experiment, noise is inevitable, which leads to a certain distance between the prepared QECC and the target QECC.  In this context, what kind of impact will this have on the distances of QECCs?   Let $\cQ$ be the target $((n,K,\delta))_d$ QECC.  For the perturbed QECC $\cQ'$, we use the Rains unitary enumerators to quantify how closely the distance of $\cQ'$ approaches that of the target QECC $\cQ$.
We say that the distance of a QECC $Q'$ is \emph{$\epsilon$-close} to the distance of $\cQ$ if $|KB_{\delta-1}'-A'_{\delta-1}|<\epsilon$.
In the $n$-qubit parallelized SWAP test, we need to prepare two copies of $\rho_{\cQ}$.
Due to the noise, we may prepare $\rho_{\cQ'}$,
where the trace distance between
$\rho_{\cQ'}$  and $\rho_{\cQ}$ satisfies $D(\rho_{\cQ'},\rho_{\cQ})=\frac{1}{2}\tr\left[\sqrt{(\rho_{\cQ'}-\rho_{\cQ})^{\dagger}(\rho_{\cQ'}-\rho_{\cQ})}\right]\leq\epsilon$.

\begin{proposition}\label{proposition: robustness}
  If $\cQ$ is an $((n,K,\delta))_d$ QECC, and $D(\rho_{\cQ'},\rho_{\cQ})\leq\epsilon$, then the distance of $\cQ'$ is  $\left[4(K+1)\binom{n}{\delta-1}\epsilon\right]$-close to the distance of $\cQ$.
\end{proposition}

Proposition~\ref{proposition: robustness} indicates that if the perturbed QECC $\cQ'$ is close to the target QECC $\cQ$, then the distance of $\cQ'$ is also close to the distance of $\cQ$.
The proof of Proposition~\ref{proposition: robustness} is given in Appendix~\ref{appendix: robustness}.
When measuring probabilities in the $n$-qubit parallelized SWAP test, some errors may  occur, which will also have a certain impact on the Rains unitary enumerators.  We do not intend to discuss this issue in this work.

\subsection{Determining the
$k$-uniformity of pure states}

The second application involves determining the $k$-uniformity of  pure states using the
$n$-qubit parallelized SWAP test. According to Lemma~\ref{lemma:k_uniform}, a pure $((n, 1, k+1))_d$ QECC corresponds to a $k$-uniform  state  in $\cH_{[n]}$. Thus,
the method for determining the $k$-uniformity of pure states is similar to the method for calculating the distances of QECCs.  Note that, we only need to consider $A'_j$, since $A'_j=B'_j$ for a pure state, and the $k$-uniformity $k=\max\{k'\mid A_{k'}-\binom{n}{k'}d^{-k'}=0\}$.
Given a $4$-qubit pure state $\rho_4=\ketbra{\psi}{\psi}$, assume that we obtain the probability distribution from the $4$-qubit parallelized SWAP test of $\rho_4$ and $\rho_4$: $p(0)=\frac{9}{16}$, $p(1)=0$, $p(2)=\frac{1}{16}$, $p(3)=0$, and $p(4)=\frac{1}{16}$. By calculating the Rains unitary enumerators, we have $A_0'=1$, $A_1'=2$, $A_2'=3$, $A_3'=2$, and $A_4'=1$. Since $A'_1-2=0$, and $A_2'- \frac{3}{2}\neq 0$, we can determine that $\ket{\psi}$ is a $1$-uniform state. Actually, $\ket{\psi}$ is chosen from the $4$-qubit $\ghz$ state.

Note that one may think of directly determining the $k$-uniformity of a pure state by calculating the local purity of the pure state on each $s$ subsystems, such as using the standard SWAP test \cite{barenco1997stabilization,buhrman2001quantum,garcia2013swap,cincio2018learning,harrow2013testing}, two two-body gates \cite{nakazato2012measurement}, and quantum overlapping tomography \cite{cotler2020quantum,yang2023experimental,hansenne2024optimal,wei2024optimal}. However, if the local purity is not equal to $\frac{1}{d^s}$, one need to compute the local purity on each $s-1$ subsystems, and so on. This method would consume a lot of resources. In contrast, a single $n$-qubit parallel SWAP test is sufficient to calculate all the local purities.

\subsection{Evaluating multipartite entanglement measures}
The third application involves evaluating multipartite entanglement measures using the
$n$-qubit parallelized SWAP test. Let $\ket{\psi}$ be a pure state in $\cH_{[n]}$, and $\cS=\{S_1, S_2, \cdots, S_m\}$
be a collection of subsets of $[n]$ with $S_i\subseteq [n]$ for $1\leq i\leq m$  ($\cS\neq \{\emptyset\}$). Then we define the \emph{fixed partition measure} on $\cS$:
\begin{equation}
    E_{\cS}(\ket{\psi})=1-\frac{1}{m}\sum_{i=1}^m\tr{(\rho_{S_i}^2)},
\end{equation}
where $\rho=\ketbra{\psi}{\psi}$. Note that  $E_{\cS}(\ket{\psi})=0$ when $\ket{\psi}$ is a product state, i.e. $\ket{\psi}=\bigotimes_{i=1}^n\ket{\varphi_i}$. Since the local purity $\tr{(\rho_{S}^2)}$   cannot decrease on average under local operations
and classical communications (LOCC) for $S\subseteq [n]$ \cite{beckey2021computable}, the fixed partition measure is nonincreasing on
average under LOCC. These two facts show that  $E_{\cS}(\ket{\psi})$ is a well-defined pure state entanglement measure. However, $E_{\cS}$ is  unfaithful for some $\cS$. For example, let $n=3$, $\cS=\{\{1\}\}$, $\ket{\psi}=\frac{1}{\sqrt{2}}\ket{0}(\ket{00}+\ket{11})$, then  $E_{\cS}(\ket{\psi})=0$ for the entangled state $\ket{\psi}$.

The fixed partition measure $E_{\cS}(\ket{\psi})$ quantifies $\ket{\psi}$'s entanglement across all the bipartitions $S\mid S^c$, with $S\in \cS$. When  $\ket{\psi}$ is maximally entangled across  the bipartition $S\mid S^c$ for all $S\in \cS$,  $E_{\cS}(\ket{\psi})$ obtains the maximum value. When $\cS$ consists of all the $1$-subsets of $[n]$, the fixed partition measure is the average subsystem linear entropy \cite{brennen2003observable}; when $\cS$ consists of all the $k$-subsets of $[n]$ with $k\leq \fl{n}{2}$, the fixed partition measure is the (unnormalized) Scott measure \cite{scott2004multipartite}; and when $T\subseteq [n]$ and $\cS$ consists of all the subsets of $T$, the fixed partition measure is the concentratable entanglement \cite{beckey2021computable}. According to Proposition~\ref{proposition: overlap}, the fixed partition measure can be computed from the $n$-qubit parallelized  SWAP test.
\begin{proposition}\label{proposition: fixedmeasure}
 Let $\ket{\psi}$ be a pure state in $\cH_{[n]}$.  In the $n$-qubit parallelized SWAP test of $\ketbra{\psi}{\psi}$ and $\ketbra{\psi}{\psi}$,  we have
 \begin{equation}
     E_{\cS}(\ket{\psi})=1-\frac{1}{m}\sum_{i=1}^m\sum_{\bz\in \bbZ_2^n}(-1)^{\bz\cdot \textbf{t}_i}p(\bz),
   \end{equation}
where $\cS=\{S_1, S_2,\cdots, S_m\}$ with $S_i\subseteq [n]$, $\bt_i\in \bbZ_2^n$,  and $\supp(\bt_i)=S_i$ for $1\leq i\leq m$.
\end{proposition}

Proposition~\ref{proposition: fixedmeasure} generalizes the results in \cite{beckey2021computable}, which shows that the concentratable entanglement can be computed from the $n$-qubit parallelized SWAP test.

\section{Conclusion}\label{sec:conclusion}
We established a direct connection between quantum weight enumerators and the
$n$-qubit parallelized SWAP test. Specifically, we demonstrated that the shadow enumerators correspond to probabilities observed in the
$n$-qubit parallelized SWAP test. Furthermore, we showed that both the Shor-Laflamme weight enumerators and the Rains unitary enumerators can be derived from this test. Leveraging these relationships, we developed a method to determine the distances of QECCs and the
$k$-uniformity of pure states. Additionally, we applied the
$n$-qubit parallelized SWAP test to compute multipartite entanglement measures.

 The connection between quantum weight enumerators and the
$n$-qubit parallelized SWAP test offers a novel perspective on quantum weight enumerators, and provides a powerful tool for characterizing and studying qLDPC codes, surface codes, and related QECCs. Since the SWAP test is a well-established quantum procedure, our method presents a practical approach for experimentally computing the distances of QECCs. Moreover, our quantum circuit-based proof of the shadow inequalities can be extended to generalized shadow inequalities. As generalized shadow inequalities involve multiple copies of states \cite{rains2000polynomial}, their circuit implementation may rely on the parallelized permutation test \cite{liu2024generalized}.

\textit{\textbf{Note added}}\textbf{---}Two months after uploading our manuscript to arXiv, we became aware that Ref.~\cite{miller2024experimental} independently provided a physical interpretation of shadow enumerators as triplet probabilities in two-copy Bell measurements, which differs from our interpretation of shadow enumerators as probabilities in the SWAP test.   In addition, compared to our work, Ref.~\cite{miller2024experimental} includes an analysis of the sample complexity required to achieve a given estimation accuracy.

\section{Acknowledgement}
The authors are very grateful to the associate editor and the
anonymous reviewers for providing many useful suggestions which have greatly improved the presentation
of our paper. The authors thank Felix Huber, Yu Ning, Zuo Ye, Wenjun Yu, Xingjian Zhang, Tianfeng Feng, Jue Xu, and Yue Wang for their helpful discussion and suggestions.

\appendix
\subsection{Two equivalent definitions of shadow enumerators.}\label{appendix:lemma_equivalent}

\begin{lemma}\label{lemma:equivalent}
  Let $\rho$ be an $n$-qubit state, then
  \begin{equation}
    \sum_{\supp(E_{\alpha})=T,\, E_{\alpha}\in \cE_2}\tr(\rho E_{\alpha} \widetilde{\rho}E_{\alpha})= \sum_{S\subseteq [n]}(-1)^{|S\cap T^c|}\tr(\rho_{S}^2).
  \end{equation}
\end{lemma}

\begin{proof}
We can also express $\widetilde{\rho}$ as \cite{wyderka2018constraints,eltschka2018distribution}:
\begin{equation}
 \widetilde{\rho}=\sum_{S\subseteq [n]}(-1)^{|S|}\rho_S\otimes I_{S^c},
\end{equation}
where $I_{S^c}$ is the identity matrix on $\otimes_{i\in S^c}\cH_i$.
  According to Observation~1 of Ref.~\cite{huber2018bounds}, we have
    \begin{equation}
  \sum_{\supp(E_{\alpha})\subseteq T,\, E_{\alpha}\in \cE_2}E_{\alpha}\rho E_{\alpha}=2^{|T|}\rho_{T^c}\otimes I_{T}.
    \end{equation}
By using the inclusion-exclusion principle, we have
\begin{equation}
      \sum_{\supp(E_{\alpha})=T,\, E_{\alpha}\in \cE_2}E_{\alpha}\rho E_{\alpha}=\sum_{R\subseteq T}(-1)^{|T\setminus R|}2^{|R|}\rho_{R^c}\otimes I_{R}.
\end{equation}
Next,
\begin{equation}
\begin{aligned}
   &\sum_{\supp(E_{\alpha})=T,\, E_{\alpha}\in \cE_2}\tr(\rho E_{\alpha} \widetilde{\rho}E_{\alpha})\\=&  \tr\left[\left(\sum_{\supp(E_{\alpha})=T,\, E_{\alpha}\in \cE_2}E_{\alpha}\rho E_{\alpha}\right) \widetilde{\rho}\right]\\
   =&\tr\left[\left(\sum_{R\subseteq T}(-1)^{|T\setminus R|}2^{|R|}\rho_{R^c}\otimes I_{R}\right)\right.\\&\left.\times \left(\sum_{S\subseteq [n]}(-1)^{|S|}\rho_{S}\otimes I_{S^c}\right)\right]\\
   =&\sum_{R\subseteq T}\sum_{S\subseteq [n]}(-1)^{|T\setminus R|+|S|}2^{|R|}\tr\left(\rho_{R^c}\otimes I_{R}\cdot \rho_{S}\otimes I_{S^c}\right)\\
   =&\sum_{R\subseteq T}\sum_{S\subseteq [n]}(-1)^{|T\setminus R|+|S|}2^{|R|+|R\setminus S|}\tr\left(\rho_{R^{c}\cap S}^2\right).
\end{aligned}
\end{equation}
\begin{figure}[t]
		\centering
		\includegraphics[scale=0.4]{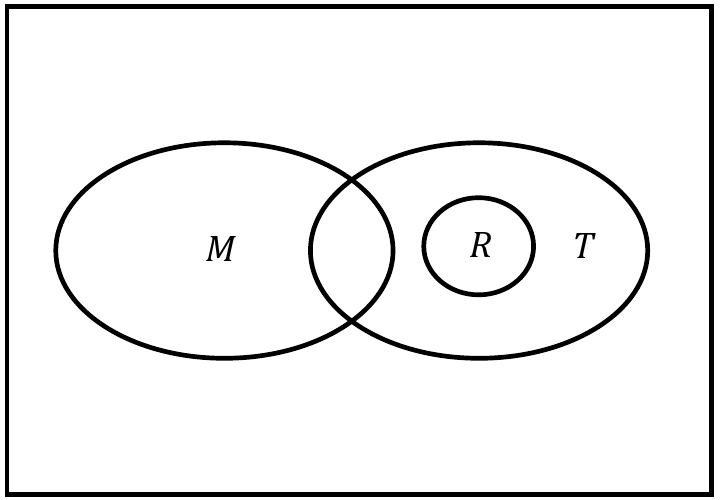}
		\caption{The Venn diagram  for the proof of Lemma~\ref{lemma:equivalent}. } \label{fig:Venn}
\end{figure}
Assume $R^{c}\cap S=M$, then we only need to show that
\begin{equation}
  \sum_{R\subseteq T}\sum_{S\subseteq [n]}\sum_{R^{c}\cap S=M}(-1)^{|T\setminus R|+|S|}2^{|R|+|R\setminus S|}=(-1)^{|M\cap T^c|}.
\end{equation}
Note that $M$ and $T$ are fixed. If $R\subseteq T$,  and $R^{c}\cap S=M$,  then  according to Fig.~\ref{fig:Venn}, we have $R\subseteq T\setminus M$. Similarly, if $S\subseteq [n]$ and $R^{c}\cap S=M$, then we have   $M\subseteq S\subseteq M\cup R$.
Thus we obtain that
 \begin{equation*}
 \begin{aligned}
   &\sum_{R\subseteq T}\sum_{S\subseteq [n]}\sum_{R^{c}\cap S=M}(-1)^{|T\setminus R|+|S|}2^{|R|+|R\setminus S|}\\&=\sum_{R\subseteq T\setminus M}\sum_{M\subseteq S\subseteq M\cup R}(-1)^{|T\setminus R|+|S|}2^{|R|+|R\setminus S|}\\
   &=\sum_{R\subseteq T\setminus M}\sum_{i=0}^{|R|}\binom{|R|}{i}(-1)^{|T\setminus R|+|M|+i}2^{|R|+|R|-i}\\
   &=(-1)^{|M|}\sum_{R\subseteq T\setminus M}(-1)^{|T\setminus R|}2^{|R|}\sum_{i=0}^{|R|}\binom{|R|}{i}(-1)^i2^{|R|-i}\\
    &=(-1)^{|M|}\sum_{R\subseteq T\setminus M}(-1)^{|T\setminus R|}2^{|R|}(2-1)^{|R|}\\
     &=(-1)^{|M|}\sum_{i=0}^{|T\setminus M|}\binom{|T\setminus M|}{i}(-1)^{|T|-i}2^{i}\\
                   \end{aligned}
 \end{equation*}
     \begin{equation}
 \begin{aligned}
     &=(-1)^{|M|+|T|}\sum_{i=0}^{|T\setminus M|}\binom{|T\setminus M|}{i}(-2)^{i}\\
     &=(-1)^{|M|+|T|+|T\setminus M|}\\
     &=(-1)^{|M\cap T^c|}.
\end{aligned}
 \end{equation}
\end{proof}

\subsection{Monogamy inequalities from the shadow inequalities}\label{appendix: monogamy}
Firstly, we need to show a lemma.

\begin{lemma}\label{lemma: sumzero}
Let $\bu\in \bbZ_2^n$, then

\begin{equation}
\sum_{\bz\in\bbZ_2^n}(-1)^{\bz\cdot \bu}=\begin{cases}
0,   &\text{if} \  \bu\neq \mathbf{0}; \\
 2^n,   &\text{if} \ \bu= \mathbf{0}. \\
\end{cases}
\end{equation}
Alternatively,
\begin{equation}
 \sum_{S\subseteq [n]}(-1)^{|S\cap T|}=\begin{cases}
0,   &\text{if} \  T\neq \emptyset; \\
 2^n,   &\text{if} \ T= \emptyset, \\
\end{cases}
\end{equation}
\end{lemma}
where $T\subseteq [n]$.

\begin{proof}
    If $\bu=\mathbf{0}$, then it is easy to see that   $\sum_{\bz\in\bbZ_2^n}(-1)^{\bz\cdot \bu}=2^n$.
Let $\textbf{e}_i\in \bbZ_2^n$ with $\supp(\textbf{e}_i)=\{i\}$ for $1\leq i\leq n$. If $\bu\neq \mathbf{0}$, then there exists an $i\in \supp(\textbf{u})$ such that $(-1)^{\textbf{e}_i\cdot \bu}=-1$. We have
\begin{equation}
\begin{aligned}
&-\sum_{\bz\in \bbZ_2^n}(-1)^{\bz\cdot \bu}=(-1)^{\textbf{e}_i\cdot \bu}\sum_{\bz\in \bbZ_2^n}(-1)^{\bz\cdot \bu}\\&=\sum_{\bz\in \bbZ_2^n}(-1)^{(\textbf{e}_i+\bz)\cdot \bu}=\sum_{\bz\in \bbZ_2^n}(-1)^{\bz\cdot \bu}.
\end{aligned}
\end{equation}
Therefore, $\sum_{\bz\in \bbZ_2^n}(-1)^{\bz\cdot \bu}=0$ if $\bu\neq \mathbf{0}$.
\end{proof}

Assume $\ket{\psi}\in \cH_{[n]}$. The concurrence acrross the bipartition $S|S^c$ is defined as  $C_{S|S^c}(\ket{\psi})=\sqrt{2[1-\tr(\rho_S^2)]}$ \cite{wootters1998entanglement}, where $\rho=\ketbra{\psi}{\psi}$.
Let $T\neq \emptyset$, then
\begin{equation}
\begin{aligned}
   &\sum_{ S\subseteq [n]}(-1)^{|S\cap T|}C_{S|S^c}^2(\ket{\psi})=   \sum_{ S\subseteq [n]}(-1)^{|S\cap T|}2[1-\tr(\rho_S^2)]\\&= -2\sum_{ S\subseteq [n]}(-1)^{|S\cap T|}\tr(\rho_S^2)\leq 0.
\end{aligned}
\end{equation}
where the second equation is obtained from Lemma~\ref{lemma: sumzero}, and the last inequality is from the shadow inequalities.

Therefore, the shadow inequalities give $2^{n}-1$ monogamy inequalities \cite{eltschka2018exponentially,huber2021positive},
\begin{equation}
   \sum_{ S\subseteq [n]}(-1)^{|S\cap T|}C_{S|S^c}^2(\ket{\psi})\leq 0,   \quad  \emptyset\neq T\subseteq [n].
\end{equation}
Note that $C_{S|S^c}(\ket{\psi})=C_{S^c|S}(\ket{\psi})$. Moreover, when $|T|$ is odd, $(-1)^{|S\cap T|}+(-1)^{|S^c\cap T|}=0$. Then $\sum_{S\subseteq [n]}(-1)^{|S\cap T|}C_{S|S^c}^2(\ket{\psi})= 0$,
when $|T|$ is odd. Thus, there are  $2^{n-1}-1$ real monogamy inequalities,
\begin{equation}
   \sum_{ S\subseteq [n]}(-1)^{|S\cap T|}C_{S|S^c}^2(\ket{\psi})\leq 0,   \quad    |T| \ \text{is even}, \ \emptyset\neq T\subseteq [n].
\end{equation}

\subsection{The proof of Proposition~\ref{proposition: pz}}\label{appendix: proposition 1}

\begin{proof}
Firstly, we need to show a lemma.
\begin{lemma}[\cite{rains2000polynomial}]\label{lemma:trace}
Assume  $\rho$ and $\sigma$ are two states on $\cH_{[n]}$. Let  $s_k$ be a swap operator that swaps the $k$-th subsystem of $\rho$ and the $k$-th subsystem of $\sigma$, and $(a_1,a_2,\ldots,a_n)\in \bbZ_2^n$, then
\begin{equation}
    \tr\left[(s_1^{a_1}\otimes s_2^{a_2}\otimes \cdots\otimes s_n^{a_n})(\rho\otimes \sigma)\right]=\tr(\rho_S\sigma_S),
\end{equation}
where $S=\supp(a_1, a_2, \ldots, a_n)$.
\end{lemma}

Now, we begin to prove Propositions~\ref{proposition: pz}.
Assume that the spectral decompositions of $\rho$ and $\sigma$ are $\rho=\sum_{i}p_i\ketbra{\psi_i}{\psi_i}$ and $\sigma=\sum_{j}q_j\ketbra{\phi_j}{\phi_j}$ respectively, where $0\leq p_i, q_j\leq 1$, $\sum_{i}p_i=1$, and $\sum_{j}q_j=1$.
Then $\rho\otimes \sigma=\sum_{i,j}p_iq_j\ket{\psi_i}\ket{\phi_j}\bra{\psi_i}\bra{\phi_j}$. In the $n$-qubit parallelized SWAP test, we only need to consider the ensemble $\{p_iq_j, \ket{\psi_i}\ket{\phi_j}\}$. Through the $n$-qubit parallelized SWAP test, we have
\begin{equation*}
\begin{aligned}
&\ket{0}\ket{0}\cdots \ket{0} \ket{\psi_i}\ket{\phi_j}\xrightarrow{H} \frac{1}{2^\frac{n}{2}}\sum_{(a_1,a_2,\ldots ,a_n)\in \bbZ_2^n}\ket{a_1}\ket{a_2}\cdots\\&\ket{a_n} \ket{\psi_i}\ket{\phi_j}
\xrightarrow{\text{control \ SWAP}} \frac{1}{2^\frac{n}{2}}\sum_{(a_1,a_2,\ldots ,a_n)\in \bbZ_2^n}\ket{a_1}\ket{a_2}\\&\cdots\ket{a_n}\left[(s_1^{a_1}\otimes s_2^{a_2}\otimes\cdots\otimes s_n^{a_n})\ket{\psi_i}\ket{\phi_j}\right]
\xrightarrow{H} \frac{1}{2^n}\times\\&\sum_{(b_1,b_2,\ldots ,b_n)\in \bbZ_2^n}\sum_{(a_1,a_2,\ldots ,a_n)\in \bbZ_2^n}(-1)^{\sum_{i=1}^na_ib_i}\ket{b_1}\ket{b_2}\cdots\\&\ket{b_n}\left[(s_1^{a_1}\otimes s_2^{a_2}\otimes\cdots\otimes s_n^{a_n})\ket{\psi_i}\ket{\phi_j}\right]=\ket{\Psi_{i,j}},
\end{aligned}
\end{equation*}
where $s_k$ is a swap operator that swaps the $k$-th subsystem of $\ket{\psi_i}$ and the $k$-th subsystem of $\ket{\phi_j}$  for $1\leq k\leq n$.

Let $\bz=(b_1,b_2,\cdots, b_n)\in \bbZ_2^n$, then
the probability
\begin{equation*}
\begin{aligned}
&p(\bz)=\sum_{i,j}p_iq_j\bra{b_1}\bra{b_2}\cdots\bra{b_n}\left(\tr_{B_1B_2\cdots B_n C_1C_2\cdots C_n}\ket{\Psi_{i,j}}\right.\\&\left.\bra{\Psi_{i,j}}\right)\ket{b_1}\ket{b_2}\cdots\ket{b_n}
=\sum_{i,j}p_iq_j\frac{1}{2^n}\left[\sum_{(a_1,a_2,\ldots ,a_n)\in \bbZ_2^n}\right.\\&\left.(-1)^{\sum_{i=1}^na_ib_i}\bra{\psi_i}\bra{\phi_j}(s_1^{a_1}\otimes s_2^{a_2}\otimes\cdots\otimes s_n^{a_n})\right]\times
\\& \frac{1}{2^n}\left[\sum_{(c_1,c_2,\ldots ,c_n)\in \bbZ_2^n}(-1)^{\sum_{i=1}^nc_ib_i}(s_1^{c_1}\otimes s_2^{c_2}\otimes\cdots\otimes s_n^{c_n})\right.\\&\left.\ket{\psi_i}\ket{\phi_j}\right]=\frac{1}{4^n}\sum_{(a_1,a_2,\ldots ,a_n)\in \bbZ_2^n}\sum_{(c_1,c_2,\ldots ,c_n)\in \bbZ_2^n}\\&(-1)^{\sum_{i=1}^n(a_i+c_i)b_i}\tr\left[(s_1^{a_1+c_1}\otimes s_2^{a_2+c_2}\otimes\cdots\otimes s_n^{a_n+c_n})\right.\\&\left.\sum_{i,j}p_iq_j\ket{\psi_i}\ket{\phi_j}\bra{\psi_i}\bra{\phi_j}\right]=\frac{1}{2^n}\sum_{(d_1,d_2,\ldots ,d_n)\in \bbZ_2^n}\\&(-1)^{\sum_{i=1}^nd_ib_i}\tr\left[(s_1^{d_1}\otimes s_2^{d_2}\otimes\cdots\otimes s_n^{d_n})(\rho\otimes \sigma)\right]\\&\xlongequal[]{\supp(d_1,d_2,\cdots,d_n)=S}\frac{1}{2^n}\sum_{S\subseteq [n]}(-1)^{|S\cap T|}\tr(\rho_S\sigma_S),
\end{aligned}
\end{equation*}
 where the last equation is obtained from Lemma~\ref{lemma:trace}.
\end{proof}

\subsection{The proof of Proposition~\ref{proposition: overlap}}\label{appendix: proposition 2}
\begin{proof}
Firstly, we need to prove a fact: assume $S, T\subseteq [n]$, then
\begin{equation}\label{eq: ST}
 \sum_{R\subseteq [n]}(-1)^{|R\cap T|+|S\cap R|}=  \left\{
\begin{array}{ll}
 0,   &\text{if} \ S\neq T; \\
 2^n,   &\text{if} \ S= T. \\
\end{array}
 \right.
\end{equation}
The proof of this fact is as follows.  If $S= T$, then  $\sum_{R\subseteq [n]}1=2^n$. We only need to consider the case  $S\neq T$.  There exist two unique vectors $\bu, \bv\in \bbZ_2^n$, such that $\supp(\bu)=S$ and $\supp(\bv)=T$. Then
\begin{equation}
\begin{aligned}
 &\sum_{R\subseteq [n]}(-1)^{|R\cap T|+|S\cap R|}=\sum_{\bz\in \bbZ_2^n}(-1)^{\bz\cdot \bv+ \bu\cdot \bz}\\&= \sum_{\bz\in \bbZ_2^n}(-1)^{\bz\cdot (\bu+\bv)}=\sum_{\bz\in \bbZ_2^n}(-1)^{\bz\cdot \textbf{w}},
 \end{aligned}
\end{equation}
where $\textbf{w}=\bu+\bv$. Since $S\neq T$, we have $\textbf{w}\neq \mathbf{0}$. By Lemma~\ref{lemma: sumzero}, we have $\sum_{\bz\in \bbZ_2^n}(-1)^{\bz\cdot \textbf{w}}=0$. Thus  $\sum_{R\subseteq [n]}(-1)^{|R\cap T|+|S\cap R|}=0$ for $S\neq T$.

Now, we can prove this proposition, that is
\begin{equation}
\begin{aligned}
    &\sum_{\bz\in \bbZ_2^n}(-1)^{\bz\cdot \textbf{t}}p(\bz)=\sum_{R\subseteq [n]}(-1)^{|R\cap T|}\left[\frac{1}{2^n}\sum_{S\subseteq [n]}(-1)^{|S\cap R|}\right.\\&\left.\tr(\rho_S\sigma_S)\right]
    =\frac{1}{2^n}\sum_{S\subseteq [n]}\sum_{R\subseteq [n]}(-1)^{|R\cap T|+|S\cap R|}\tr(\rho_S\sigma_S)\\
    &=\frac{1}{2^n}\sum_{S=T}\sum_{R\subseteq [n]}(-1)^{|R\cap T|+|S\cap R|}\tr(\rho_S\sigma_S)=\tr(\rho_T\sigma_T).
\end{aligned}
\end{equation}
\end{proof}

\subsection{Characterizing $\rho$ and $\sigma$  from the $n$-qubit parallelized SWAP test of $\rho$ and $\sigma$}\label{appendix:charaterizing}
In the $n$-qubit parallelized SWAP test of $\rho$ and $\sigma$,
 \begin{enumerate}[(i)]
     \item   $p(\mathbf{0})=1$ if and only if $\rho=\sigma=\bigotimes_{i=1}^n\ketbra{\psi_i}{\psi_i}$;
 \item  $p(\bz)=\frac{1}{2^n}$ for all $\bz\in \bbZ_2^n$ if and only if $\rho_T\sigma_T=\mathbf{0}$ for all $\emptyset\neq T\subseteq [n]$.
 \end{enumerate}

\begin{proof}
    (i) If $\rho=\sigma=\bigotimes_{i=1}^n\ketbra{\psi_i}{\psi_i}$, then $\tr(\rho_{T}\sigma_{T})=1$ for all $T\subseteq [n]$. Thus $p(\mathbf{0})=1$.

    If $p(\mathbf{0})=1$, then $p(\bz)=0$ for $\bz\neq \mathbf{0}$. According to Proposition~\ref{proposition: overlap}, we have $\tr(\rho_T\sigma_T)=1$ for all $T\subseteq [n]$.
    Assume $\rho=\sum_{i}p_i\ketbra{\psi_i}{\psi_i}$ and $\sigma=\sum_{j}q_j\ketbra{\phi_j}{\phi_j}$, where $0\leq p_i, q_j\leq 1$, $\sum_{i}p_i=1$, and $\sum_{j}q_j=1$. We have
\begin{equation}\label{eq: trrhosigma}
\tr(\rho\sigma)=\sum_{i,j}p_iq_j|\braket{\psi_i}{\phi_j}|^2=1,
\end{equation}
then we obtain $|\braket{\psi_i}{\phi_j}|=1$ for all $i,j$, i.e., $\ket{\psi_i}=w_{i,j}\ket{\phi_j}$, where $w_{i,j}\in \bbC$, and $|w_{i,j}|=1$ for all $i,j$. This implies that $\rho=\sigma$, and $\rho$ is a pure state. Moreover, we obtain that
 $\rho_T=\sigma_T$, and $\rho_T$ is a pure state for all $T\subseteq [n]$. Thus there exists  $\ket{\psi_i}\in \cH_i$ for $1\leq i\leq n$, such that $\rho=\sigma=\bigotimes_{i=1}^n\ketbra{\psi_i}{\psi_i}$.

(ii) If $\rho_T\sigma_T=\mathbf{0}$ for all  $\emptyset\neq T\subseteq [n]$, then $\tr(\rho_T\sigma_T)=0$ for all $\emptyset\neq T\subseteq [n]$. Note that $\tr(\rho_T\sigma_T)=1$ if $T=\emptyset$. Thus $p(\bz)=\frac{1}{2^n}$ for all $\bz\in \bbZ_2^n$.

 If $p(\bz)=\frac{1}{2^n}$ for all $\bz\in \bbZ_2^n$, then
 $\tr(\rho_T\sigma_T)=\frac{1}{2^n}\sum_{\bz\in \bbZ_2^n}(-1)^{\bz\cdot \bt}$, where $\bt\in\bbZ_2^n$, and $\supp(\bt)=T$. By Lemma~\ref{lemma: sumzero}, we obtain that $\tr(\rho_T\sigma_T)=0$ for all $\emptyset\neq T\subseteq [n]$. By Eq.~\eqref{eq: trrhosigma}, we know that if $\tr(\rho\sigma)=0$, then $\rho\sigma=\mathbf{0}$. Thus, $\rho_T\sigma_T=\mathbf{0}$ for all $\emptyset\neq T\subseteq [n]$.
\end{proof}

\subsection{The complexity of the algorithm for calculating the distances of QECCs}\label{appendix:complexity}

To analyze the complexity of code distance determination algorithm, we first analyze the sample complexity.
Consider general random variable $R= \sum_{\mathbf{z}\in \mathbb{Z}_{2}^{n}}f(\mathbf{z})p(\mathbf{z})$ and its estimator after sampling $N$ times: $\hat{R}= \sum_{z\in \mathbb{Z}_{2}^{n}} f(\mathbf{z}) \frac{N_{\mathbf{z}}}{N}$, where $N_{\mathbf z}$ is the number of occurrences of $\mathbf z$ in the sample. Since the $N_{\mathbf{z}}$ forms a multinomial distribution with $2^{n}$ terms and probability $p_{\mathbf{z}},\mathbf{z}\in \mathbb{Z}_{2}^{n}$, $\mathbb{E}N_{\mathbf{z}}=Np(\mathbf{z})$ and thus $\mathbb E \hat{R}=R$, the estimator is unbiased. We calculate the variance of the estimator:
\begin{equation}
\begin{aligned}
\mathrm{Var} \hat{R}=& \sum_{\mathbf{z},\mathbf{w}\in \mathbb{Z}_{2}^{n}} \frac{f(\mathbf{z})f(\mathbf{w})}{N^{2}} \mathrm{Cov}(N_{\mathbf{z}},N_{\mathbf{w}})
\\=&\sum_{\mathbf{z},\mathbf{w}\in \mathbb{Z}_{2}^{n}} \frac{f(\mathbf{z})f(\mathbf{w})}{N} (\delta_{\mathbf{z},\mathbf{w}} p(\mathbf{z}) - p(\mathbf{z})p(\mathbf{w}))
\\=& \frac{1}{N} \sum_{\mathbf{z}\in \mathbb{Z}_{2}^{n}} f(\mathbf{z})^{2}p(\mathbf{z}) - \frac{1}{N} R^{2}\label{appendix:variance-estimator},
\end{aligned}
\end{equation}
where we used the property the covariance of multinomial distribution $\mathrm{Cov}(N_{\mathbf{z}},N_{\mathbf{z}'})=N(\delta_{\mathbf{z},\mathbf{z}'}p(\mathbf{z}) -p(\mathbf{z})p(\mathbf{z}') )$. Denote the combinatorial number $f_{k}(w):=\sum_{|\mathbf{t}|=k,\mathbf{t}\in \mathbb{Z}_{2}^{n}}(-1)^{\mathbf{z}\cdot \mathbf{t}}=\sum_{m=0}^{\mathrm{min}(k,w)}(-1)^{m}{n-w \choose k-m}{w\choose m}$ where $\mathrm{wt}(\mathbf{z})=w$. We can rewrite $B_{j}'(\rho,\rho)$ and $A_{j}'(\rho,\rho)$ in terms of $p(\mathbf{z})$:
\begin{equation}
\begin{aligned}
A_{j}'(\rho,\rho)  =&\sum_{|T|=j}\mathrm{Tr}\left[\rho_{T}^{2}\right]\\=&\sum_{\mathrm{wt}(\mathbf{t})=j}\sum_{\mathbf{z}\in \mathbb{Z}_{2}^{n}}(-1)^{\mathbf{z}\cdot \mathbf{t}}p(\mathbf{z})\\=&\sum_{\mathbf{z}\in \mathbb{Z}_{2}^{n}}f_{j}(\mathrm{wt}(\mathbf{z}))p(\mathbf{z}), \\
 B_{j}'(\rho,\rho)   =& A_{n-j}'(\rho,\rho)\\=&\sum_{z\in \mathbb{Z}_{2}^{n}}f_{n-j}(\mathrm{wt}(\mathbf{z}))p(\mathbf{z}),\\
K B_{j}'(\rho,\rho) - A_{j}'(\rho,\rho)   =& \sum_{\mathbf{z}\in  \mathbb{Z}_{2}^{n}} \left[Kf_{n-j}(\mathrm{wt}(\mathbf{z}))\right.\\ &\left.- f_{j}(\mathrm{wt}(\mathbf{z}))\right] p(\mathbf{z}).
\end{aligned}
\end{equation}
Using Eq.~\eqref{appendix:variance-estimator}, the corresponding unbiased estimator
\begin{equation}
\begin{aligned}
&K\hat{B}_{j}'(\rho,\rho)-\hat{A}_{j}'(\rho,\rho) \\&= \sum_{\mathbf{z}\in \mathbb{Z}_{2}^{n}} (Kf_{n-j}(\mathrm{wt}(\mathbf{z})) - f_{j}(\mathrm{wt}(\mathbf{z}))) \frac{N_{\mathbf{z}}}{N}
\end{aligned}
\end{equation}
has variance
\begin{equation*}
\begin{aligned}
&\mathrm{Var}[K\hat{B}_{j}'(\rho,\rho)-\hat{A}_{j}'(\rho,\rho)] \\=&\frac{1}{N}\left( \sum_{z\in \mathbb{Z}_{2}^{n}} [Kf_{n-j}(\mathrm{wt}(\mathbf{z}))- f_{j}(\mathrm{wt}(\mathbf{z}))]^{2}p(\mathbf{z})\right.\\&\left.- \underbrace{ (KB_{j}'(\rho,\rho)-A_{j}'(\rho,\rho))^{2} }_{ 0 }\right)\\
\end{aligned}
\end{equation*}
\begin{equation}
\begin{aligned}
=&\frac{1}{N}\left( \sum_{z\in \mathbb{Z}_{2}^{n}} [Kf_{n-j}(\mathrm{wt}(\mathbf{z}))- f_{j}(\mathrm{wt}(\mathbf{z}))]^{2}p(\mathbf{z}) \right).
\end{aligned}
\end{equation}

If the code has distance at least $j$, $K B_{j}'(\rho,\rho) - A_{j}'(\rho,\rho) $ has real value 0, we can restrict the statistical estimation of $K B_{j}'(\rho,\rho) - A_{j}'(\rho,\rho) $ to range $[0,\epsilon]$ by restricting $\sqrt{ \mathrm{Var}[K\hat{B}_{j}'-\hat{A}_{j}'] }\sim\epsilon$, which means we should sample $N\sim \sum_{z\in \mathbb{Z}_{2}^{n}} [Kf_{n-j}(\mathrm{wt}(\mathbf{z})) -f_{j}(\mathrm{wt}(\mathbf{z}))]^{2} p(\mathbf{z})/ \epsilon ^{2}$ times. Since $|f_j(w)|\leq {n\choose j}$ with equal sign gained when $w=0$, the order of sampling number is upper bounded by
\begin{equation}
N=O\left(\frac{(K+1)^2 {n\choose j}^2}{\epsilon^2}\right).
\end{equation}
One should note that since these weight enumerators can take continuous values, it is generally not possible to determine if a code satisfies $KB_j'(\rho,\rho)=A_j'(\rho,\rho)$ perfectly using a finite number of samples.

Another thing to note is that although there are $2^n$ possible outcomes for $\mathbf z\in \mathbb Z_2^n$, the classical post-processing complexity is not necessarily exponential if the number of samples is sub-exponential. For example, to estimate $R=\sum_{z\in \mathbb Z_2^n} f(\mathbf {z})p(\mathbf{ z})$ we only need to go over all samples $\mathbf z$ and sum up the corresponding coefficient $f(\mathbf z)$, then divided by $N$. For the task of computing code distances, all coefficients can be pre-computed within polynomial time.

In summary, if we want to determine the code distance to be at least $j$ to accuracy $\epsilon$, the algorithm complexity is given by
\begin{equation}
O\left(\frac{(K+1)^2 \sum_{j'=0}^{j} {n\choose j'}^{2}}{\epsilon^2}\right) = O\left(\frac{K^{2}n^{2j}}{\epsilon ^{2}}\right).
\end{equation}
Note that if the code distance is small and code space dimension is low, the algorithm can be carried out efficiently. However, for general codes with $K=O(\exp(n))$ and $j=O(n)$, it requires exponentially many samples and post-processing time to perform.

\begin{table*}[t]	
\renewcommand\arraystretch{1.4}	
	\centering
	\renewcommand\tabcolsep{6pt}
     \caption{The Shor-Laflamme weight enumerators and the Rains unitary enumerators of the five-qubit code $\cQ_5$, the Steane code $\cQ_7$ and the Shor code $\cQ_9$.}\label{Table: code}
	\begin{tabular}{|c|c|c|c|c|c|c|c|c|c|c|}
		\hline
	$A_j$  &$A_0$ &$A_1$ &$A_2$ &$A_3$ &$A_4$ &$A_5$  &$A_6$ &$A_7$ &$A_8$  &$A_9$  \\
		\hline
  	$\cQ_5$  &$1$ &$0$ &$0$ &$0$ &$15$ &$0$  &$-$ &$-$ &$-$  &$-$  \\
    $\cQ_7$  &$1$ &$0$ &$0$ &$0$ &$21$ &$0$  &$42$ &$0$ &$-$  &$-$  \\
      $\cQ_9$  &$1$ &$0$ &$9$ &$0$ &$27$ &$0$  &$75$ &$0$ &$144$  &$0$  \\
		\hline
  $B_j$  &$B_0$ &$B_1$ &$B_2$ &$B_3$ &$B_4$ &$B_5$  &$B_6$ &$B_7$ &$B_8$  &$B_9$  \\
  	\hline
  	$\cQ_5$  &$\frac{1}{2}$ &$0$ &$0$ &$15$ &$\frac{15}{2}$ &$9$  &$-$ &$-$ &$-$  &$-$  \\
    $\cQ_7$  &$\frac{1}{2}$ &$0$ &$0$ &$\frac{21}{2}$ &$\frac{21}{2}$ &$63$  &$21$ &$\frac{45}{2}$ &$-$  &$-$  \\
      $\cQ_9$  &$\frac{1}{2}$ &$0$ &$\frac{9}{2}$ &$\frac{39}{2}$ &$\frac{27}{2}$ &$\frac{207}{2}$  &$\frac{75}{2}$ &$\frac{333}{2}$ &$72$  &$\frac{189}{2}$  \\
	\hline
  $A_j'$  &$A_0'$ &$A_1'$ &$A_2'$ &$A_3'$ &$A_4'$ &$A_5'$  &$A_6'$ &$A_7'$ &$A_8'$  &$A_9'$  \\
  \hline
  	$\cQ_5$  &$1$ &$\frac{5}{2}$ &$\frac{5}{2}$ &$\frac{5}{4}$ &$\frac{5}{4}$ &$\frac{1}{2}$  &$-$ &$-$ &$-$  &$-$  \\
    $\cQ_7$  &$1$ &$\frac{7}{2}$ &$\frac{21}{4}$ &$\frac{35}{8}$ &$\frac{7}{2}$ &$\frac{21}{8}$  &$\frac{7}{4}$ &$\frac{1}{2}$ &$-$  &$-$  \\
      $\cQ_9$  &$1$ &$\frac{9}{2}$ &$\frac{45}{4}$ &$\frac{147}{8}$ &$\frac{171}{8}$ &$18$  &$\frac{93}{8}$ &$\frac{45}{8}$ &$\frac{9}{4}$  &$\frac{1}{2}$  \\
	\hline
  $B_j'$  &$B_0'$ &$B_1'$ &$B_2'$ &$B_3'$ &$B_4'$ &$B_5'$  &$B_6'$ &$B_7'$ &$B_8'$  &$B_9'$  \\
  \hline
  	$\cQ_5$  &$\frac{1}{2}$ &$\frac{5}{4}$ &$\frac{5}{4}$ &$\frac{5}{2}$ &$\frac{5}{2}$ &$1$  &$-$ &$-$ &$-$  &$-$  \\
    $\cQ_7$  &$\frac{1}{2}$ &$\frac{7}{4}$ &$\frac{21}{8}$ &$\frac{7}{2}$ &$\frac{35}{8}$ &$\frac{21}{4}$  &$\frac{7}{2}$ &$1$ &$-$  &$-$  \\
      $\cQ_9$  &$\frac{1}{2}$ &$\frac{9}{4}$ &$\frac{45}{8}$ &$\frac{93}{8}$ &$18$ &$\frac{171}{8}$  &$\frac{147}{8}$ &$\frac{45}{4}$ &$\frac{9}{2}$  &$1$  \\
 \hline
	\end{tabular}
\end{table*}

\subsection{Quantum weight enumerators of the  five-qubit code, the Steane code  and the Shor code}\label{appendix: probability}
The normalized projectors of the five-qubit code $((5,2,3))_2$, the Steane code $((7,2,3))_2$
and the Shor code $((9,2,3))_2$ are
\begin{equation}
\begin{aligned}
\rho_{\cQ_5}=&\frac{1}{2^5}(I_{2^5}+XZZXI)(I_{2^5}+IXZZX)\times\\&(I_{2^5}+XIXZZ)(I_{2^5}+ZXIXZ),\\
\rho_{\cQ_7}=&\frac{1}{2^7}(I_{2^7}+IIIXXXX)(I_{2^7}+IXXIIXX)\times\\&(I_{2^7}+XIXIXIX)(I_{2^7}+IIIZZZZ)\times\\&(I_{2^7}+IZZIIZZ)(I_{2^7}+ZIZIZIZ),\\
\rho_{\cQ_9}=&\frac{1}{2^9}(I_{2^9}+ZZIIIIIII)(I_{2^9}+ZIZIIIIII)\times\\&(I_{2^9}+IIIZZIIII)(I_{2^9}+IIIZIZIII)\times\\&(I_{2^9}+IIIIIIZZI)(I_{2^9}+IIIIIIZIZ)\times\\&(I_{2^9}+XXXXXXIII)(I_{2^9}+XXXIIIXXX),
\end{aligned}
\end{equation}
respectively, where $I_n$ is the identity matrix of order $n$, and $X, Z, I$ are the Pauli matrices. Note that the five-qubit code and the Steane code are pure codes, and the Shor code is an impure code.
In the $5$-qubit parallelized SWAP test of $\rho_{\cQ_5}$ and $\rho_{\cQ_5}$, we obtain the probability distribution (shadow enumerators): $p(0)=\frac{9}{32}$, $p(1)=\frac{3}{64}$, $p(2)=\frac{3}{64}$, $p(5)=\frac{1}{64}$, $p(\bz)=0$ for others.
In the $7$-qubit parallelized SWAP test of $\rho_{\cQ_7}$ and $\rho_{\cQ_7}$, we obtain the probability distribution: $p(0)=\frac{45}{256}$, $p(1)=\frac{3}{128}$, $p(2)=\frac{3}{128}$, $p(0010110)=p(0011001)=p(0100101)=p(0101010)=p(1000011)=p(1001100)=p(1110000)=p(1101001)=p(1100110)=p(1011010)=p(1010101)=p(0111100)=p(0110011)=p(0001111)=\frac{3}{256}$,  $p(7)=\frac{1}{256}$, $p(\bz)=0$ for others. The probability distribution of the $9$-qubit parallelized SWAP test of $\rho_{\cQ_9}$ and $\rho_{\cQ_9}$ is as follows.

Let $A=\{1, 2, 3\}$, $B=\{4, 5, 6\}$, $C=\{7, 8, 9\}$.

\begin{equation}
  p(0)=\frac{189}{1024}.
\end{equation}
\begin{equation}
  p(1)=\frac{1}{64}.
\end{equation}
\begin{equation}
p(\bz\mid \wt(\bz)=2)=\left\{
\begin{array}{ll}
    \frac{25}{1024} & \begin{subarray}{c}
         \supp(\bz)\subseteq A, \, \text{or} \, \supp(\bz)\subseteq B, \\ \text{or} \, \supp(\bz)\subseteq C; 
    \end{subarray}
    \\
   \frac{1}{256} &   \text{other cases}. \\
\end{array}
\right.
\end{equation}
\begin{equation}
p(\bz\mid \wt(\bz)=3)=\left\{
\begin{array}{ll}
    \frac{1}{64} &  \begin{subarray}{c}\supp(\bz)= A, \, \text{or} \, \supp(\bz)= B, \\ \text{or} \, \supp(\bz)= C;   \end{subarray}\\
   \frac{1}{1024} &    \begin{subarray}{c}|\supp(\bz)\cap A|=1, \,  |\supp(\bz)\cap B|=1, \\ \text{and} \, |\supp(\bz)\cap C|=1. \end{subarray}\\
\end{array}
\right.
\end{equation}
\begin{equation}
p(\bz\mid \wt(\bz)=4)=\left\{
\begin{array}{ll}
    \frac{1}{256} &  \begin{subarray}{c}A\subseteq \supp(\bz),   \, \text{or} \, B\subseteq \supp(\bz), \\ \text{or} \, C\subseteq \supp(\bz); \end{subarray}\\
   \frac{5}{1024} &   \begin{subarray}{c}|\supp(\bz)\cap A|=2,  \, |\supp(\bz)\cap B|=2, \\ \text{or} \, |\supp(\bz)\cap A|=2, \,   |\supp(\bz)\cap C|=2, \\  \text{or} \  |\supp(\bz)\cap B|=2, \,  |\supp(\bz)\cap C|=2.\end{subarray}\\
\end{array}
\right.
\end{equation}
\begin{equation}
    p(\bz\mid \wt(\bz)=5)=\frac{1}{1024}, 
\begin{subarray}{c}  A \subseteq \supp(\bz), \, |\supp(\bz)\cap B|=1, \\ |\supp(\bz)\cap C|=1, \,
\text{or} \, B\subseteq \supp(\bz), \\ |\supp(\bz)\cap A|=1, \, |\supp(\bz)\cap C|=1, \\
\text{or} \, C\subseteq \supp(\bz), \, |\supp(\bz)\cap A|=1, \\ |\supp(\bz)\cap B|=1.
\end{subarray}
\end{equation}
\begin{equation}
p(\bz\mid \wt(\bz)=6)=\left\{
\begin{array}{ll}
    \frac{1}{256} &  \begin{subarray}{c}  
    \supp(\bz)=A\cup B,   \, \text{or} \, \supp(\bz)=A\cup C, \\ \text{or} \, \supp(\bz)=B\cup C; \end{subarray}\\
   \frac{1}{1024} &    \begin{subarray}|\supp(\bz)\cap A|=2, \,   |\supp(\bz)\cap B|=2,  \\  |\supp(\bz)\cap C|=2.\end{subarray}\\
\end{array}
\right.
\end{equation}
\begin{equation}
    p(\bz\mid \wt(\bz)=7)=\frac{1}{1024}, \begin{subarray}{c} 
A\cup B\subseteq \supp(\bz), \,
\text{or} \, A\cup C\subseteq \supp(\bz),\\ 
\text{or} \, B\cup C\subseteq \supp(\bz).
\end{subarray}
\end{equation}
\begin{equation}
  p(9)=\frac{1}{1024}.
\end{equation}
\begin{equation}
  p(\bz)=0.  \quad \text{for others}
\end{equation}

The Shor-Laflamme weight enumerators and the Rains unitary enumerators of these three codes are given in Table~\ref{Table: code}.

\subsection{The proof of Proposition~\ref{proposition: robustness}}\label{appendix: robustness}
\begin{proof}
Our proof is inspired by Refs.~\cite{coles2019strong,beckey2021computable}. For simplicity, we denote $\rho=\rho_{\cQ}$, and $\rho'=\rho_{\cQ'}$.
We have
\begin{equation}
\begin{aligned}
    |\tr({\rho'}_T^2)-\tr(\rho_T^2)|=&|\tr[(\rho'_T+\rho_T)(\rho'_T-\rho_T)]| &\\
    \leq& ||\rho'_T+\rho_T||_2||\rho'_T-\rho_T||_2     \\  &(\text{Cauchy-Schwarz inequality})\\
    \leq& (||\rho'_T||_2+||\rho_T||_2)||\rho'_T-\rho_T||_2\\&  (\text{triangle inequality})   \\
     \leq& 2||\rho'_T-\rho_T||_2 \\&  (\text{the purity is less than or equal to 1})   \\
     \leq& 2||\rho'_T-\rho_T||_1    
      \\&  (\text{the monotonicity of the Schatten}\\& \text{norms})\\
     \leq& 2||\rho'-\rho||_1  \\&  (\text{the monotonicity of the trace
norm}\\& \text{under CPTP maps})\\
 =&4D(\rho',\rho)\leq 4\epsilon.  &
\end{aligned}
\end{equation}
Then
\begin{equation}
\begin{aligned}
|A'_j(\rho',\rho')-A'_j(\rho,\rho)|&=\left|\sum_{|T|=j, T\subseteq [n]}\tr({\rho'}_T^2-\rho_T^2)\right|\\
&\leq \binom{n}{j}|\tr({\rho'}_T^2)-\tr(\rho_T^2)|\\
&\leq  4\binom{n}{j}\epsilon.
\end{aligned}
\end{equation}
Moreover, $|B'_j(\rho',\rho')-B'_j(\rho,\rho)|=|A'_{n-j}(\rho',\rho')-A'_{n-j}(\rho,\rho)|\leq 4\binom{n}{j}\epsilon$.

If $\rho$ is an $((n,K,\delta))_d$ QECC, then $KB'_{\delta-1}(\rho,\rho)-A'_{\delta-1}(\rho,\rho)=0$, and
\begin{equation}
\begin{aligned}
       &|KB'_{\delta-1}(\rho',\rho')-A'_{\delta-1}(\rho',\rho')|= |K[B'_{\delta-1}(\rho',\rho')\\&-B'_{\delta-1}(\rho,\rho)]-[A'_{\delta-1}(\rho',\rho')-A'_{\delta-1}(\rho,\rho)]|\\
       &\leq K|B'_{\delta-1}(\rho',\rho')-B'_{\delta-1}(\rho,\rho)|+|A'_{\delta-1}(\rho',\rho')\\
       &-A'_{\delta-1}(\rho,\rho)|\leq 4(K+1)\binom{n}{\delta-1}\epsilon.
\end{aligned}
\end{equation}
\end{proof}
\end{sloppypar}

\bibliographystyle{IEEEtran}
\bibliography{reference}

\vspace{30pt}
\begin{IEEEbiographynophoto}{Fei Shi}
received the B.S. degree from the School of Mathematics and Information Science at Guangzhou University, China, in 2017,  and received the Ph.D.  degree from the School of Cyber Science and Technology at the University of Science and Technology of China, Hefei, in 2022. From 2022 to 2025, he was a Postdoctoral Fellow in the School of Computing and Data Science at the University of Hong Kong. He is currently an Associate Professor in the School of Computer Science and Engineering at Sun Yat-sen University. His research interests include quantum information, quantum computing, combinatorics, and their interactions.
\end{IEEEbiographynophoto}

\vspace{11pt}
\begin{IEEEbiographynophoto}{Kaiyi Guo}
 received the B.S. degree in Physics and Computer Science from Peking University in 2024. He is currently a Ph.D. candidate in the University of Hong Kong. His research interests include quantum information theory and quantum error correction.
\end{IEEEbiographynophoto}

\vspace{11pt}
\begin{IEEEbiographynophoto}{Xiande Zhang}
(Member, IEEE) received the Ph.D. degree in mathematics from Zhejiang University, Hangzhou, China, in 2009. From 2009
to 2015, she held a post-doctoral position with Nanyang Technological
University and Monash University. She is currently a Professor with the School of Mathematical Sciences, University of Science
and Technology of China. Her research interests include combinatorial
design theory, coding theory, cryptography, and quantum information
theory and their interactions. She received the 2012 Kirkman Medal
from the Institute of Combinatorics and its Applications. She is on
the editorial boards of Journal of Combinatorial Designs and IEEE
Transactions on Information Theory.
\end{IEEEbiographynophoto}

\vspace{11pt}
\begin{IEEEbiographynophoto}{Qi Zhao}
received his B.S. degree in Mathematics and the Ph.D. degree in Physics from Tsinghua University in 2014 and 2018, respectively. Then he became a postdoctoral researcher at the University of Science and Technology, China from Jan. 2019 to Oct. 2019. In Dec. 2019, he joined the University of Maryland QuICS as a Hartree Postdoctoral Fellow in quantum information science. He is currently an Assistant Professor in the School of Computing and Data Science, at the University of Hong Kong (HKU). His research interests include quantum algorithms, quantum simulation, and entanglement detection.
\end{IEEEbiographynophoto}

\vfill

\end{document}